\documentclass[11pt]{article}

\usepackage{amsmath, amssymb, graphicx, enumerate, url, wrapfig}
\usepackage{amsthm,fullpage}
\usepackage{hyperref}
\usepackage{enumitem}
\usepackage{xspace}
\usepackage{caption}
\usepackage{morefloats}

\theoremstyle{plain}
\newtheorem{theorem}{Theorem}[section]

\theoremstyle{definition}

\newtheorem{remark}[theorem]{Remark}

\title{Classic Nintendo Games are \\ (Computationally) Hard}

\author{%
	Greg Aloupis%
    \thanks{Charg\'{e} de Recherches du FNRS, D\'{e}partement d'Informatique,
    Universit\'{e} Libre de Bruxelles,
\protect\url{aloupis.greg@gmail.com}.
Work initiated while at Institute of Information Science, Academia Sinica.}
	\and
	Erik D. Demaine%
    \thanks{MIT Computer Science and Artificial Intelligence Laboratory,
      32 Vassar St., Cambridge, MA 02139, USA, \protect\url{{edemaine,aguo}@mit.edu}}
	\and
	Alan Guo%
    \footnotemark[2]
    \thanks{Partially supported by NSF grants CCF-0829672, CCF-1065125,
    and CCF-6922462.}
\and
Giovanni Viglietta%
    \thanks{School of Electrical Engineering and Computer Science, University of Ottawa, Canada, \protect\url{viglietta@gmail.com}.}
}

\newif\ifabstract
\abstracttrue
\abstractfalse
\newif\iffull
\ifabstract \fullfalse \else \fulltrue \fi

\newcounter{section-preserve}
\newcounter{theorem-preserve}
\newcommand{\blank}[1]{}
\newtoks\magicAppendix
\magicAppendix={}
\newtoks\magictoks
\newif\iflater
\laterfalse
\long\def\later#1{\magictoks={#1}%
  \addtocounter{theorem}{1}%
  \edef\magictodo{\noexpand\magicAppendix={\the\magicAppendix \par
    \the\magictoks%
  }}
  \magictodo}
\long\def\both#1{\magictoks={#1}%
  \edef\magictodo{\noexpand\magicAppendix={\the\magicAppendix \par
    \noexpand\setcounter{theorem-preserve}{\noexpand\arabic{theorem}}%
    \noexpand\setcounter{theorem}{\arabic{theorem}}%
    \noexpand\setcounter{section-preserve}{\noexpand\arabic{section}}%
    \noexpand\setcounter{section}{\arabic{section}}%
        \noexpand\let\noexpand\oldsection=\noexpand\thesection
        \noexpand\def\noexpand\thesection{\thesection}
        \noexpand\let\noexpand\oldlabel=\noexpand\label
        \noexpand\let\noexpand\label=\noexpand\blank
    \the\magictoks%
    \noexpand\setcounter{theorem}{\noexpand\arabic{theorem-preserve}}%
    \noexpand\setcounter{section}{\noexpand\arabic{section-preserve}}%
        \noexpand\let\noexpand\thesection=\noexpand\oldsection
        \noexpand\let\noexpand\label=\noexpand\oldlabel
  }}
  \magictodo
  \the\magictoks}
\def\magicappendix{\latertrue \the\magicAppendix}

\iffull
  \long\def\both#1{#1}
  \let\later=\both
  \def\magicappendix{}
\fi

\begin{document}

\maketitle

\vspace*{-0.1in}

\begin{abstract}
We prove NP-hardness results for five of Nintendo's largest video game
franchises: Mario, Donkey Kong, Legend of Zelda, Metroid, and
Pok\'emon.  Our results apply to generalized versions of Super Mario Bros.\ 1--3, The Lost
Levels, and Super Mario World; Donkey Kong Country 1--3;
all Legend of Zelda games; all
Metroid games; and all Pok\'emon role-playing games.
In addition, we prove PSPACE-completeness of the Donkey Kong Country games
and several Legend of Zelda games.
\end{abstract}

\vspace*{-0.2in}

\section{Introduction}

A series of recent papers have analyzed the computational complexity of
playing many different video games \cite{Lemmings,Platform,Games,Lemmings2},
but the most well-known classic Nintendo games have yet to be
included among these results. In this paper, we analyze some of the
best-known Nintendo games of all time: Mario, Donkey Kong, Legend of
Zelda, Metroid, and Pok\'emon.  We prove that it is NP-hard, and in some
cases PSPACE-hard, to play
generalized versions of most games in these series.
In particular, our NP-hardness results apply to
the NES games Super Mario Bros. 1--3,
Super Mario Bros.: The Lost Levels, and
Super Mario World (developed by Nintendo);
to the SNES games Donkey Kong Country~1--3 (developed by Rare Ltd.);
to all
Legend of Zelda games (developed by Nintendo);%
\footnote{We exclude the Zelda CD-i games by Philips Media,
   which Nintendo does not list as part of the Legend of Zelda series.}
%
to all Metroid games (developed by Nintendo);
and
to all Pok\'emon role-playing games
(developed by Game Freak and Creatures Inc.).%
\footnote{All products, company names, brand names, trademarks, and sprites
  are properties of their respective owners.  Sprites are used here under
  Fair Use for the educational purpose of illustrating mathematical theorems.}
Our PSPACE-hardness results apply to
to the SNES games Donkey Kong Country~1--3,
and to The Legend of Zelda: A Link to the Past.
Some of the aforementioned games are also complete for either NP or PSPACE.
All of these results are new.%
\footnote{A humorous paper
   (\url{http://www.cs.cmu.edu/~tom7/sigbovik/mariox.pdf}) and video
   (\url{http://www.youtube.com/watch?v=HhGI-GqAK9c}) by Vargomax V. Vargomax
   claims that ``generalized Super Mario Bros.\ is NP-complete'', but both
   versions have no actual proof, only nonsensical content.}

For these games, we consider the decision problem of reachability:
given a stage or dungeon, is it possible to reach the goal point $t$
from the start point~$s$?
Our results apply to
generalizations of the games where we \emph{only} generalize the map
size and leave all other mechanics of the games as they are in their
original settings. Most of our NP-hardness proofs are by reduction from
3-SAT, and 
rely on a common construction.
Similarly, our PSPACE-completeness results for Legend of Zelda:
A Link to the Past and Donkey Kong Country games are by a reduction from
True Quantified Boolean Formula (TQBF), and
rely on a common construction (inspired by a metatheorem from~\cite{Games}).
In addition, we show that several Zelda games are PSPACE-complete
by reducing from PushPush-1~\cite{PushPush}.

We can obtain some positive results if we bound the ``memory'' of the game.
For example, recall that in Super Mario Bros.\ everything substantially off
screen resets to its initial state.  Thus, if we generalize the stage size
in Super Mario Bros.\ but keep the screen size constant, then
reachability of the goal can be decided in polynomial time:
the state space is polynomial in size,
so we can simply traverse the entire state space and
check whether the goal is reachable.  Similar results hold for the other
games if we bound the screen size in Donkey Kong Country or the room
size in Legend of Zelda, Metroid, and Pok\'emon. The screen-size bound
is more realistic (though fairly large in practice), while there is no
standard size for rooms in Metroid and Pok\'emon.

\paragraph{Membership in PSPACE.}
Most of the games considered are easy to show belong to PSPACE, because every
game element's behavior is a simple (deterministic) function of the player's
moves. Therefore, we can solve a level by making moves nondeterministically
while maintaining the current game state (which is polynomial), and use that
NPSPACE = PSPACE.

Some other games, such as Legend of Zelda and its sequels, also include
enemies and other game elements that behave pseudorandomly. As long as the
random seed can be encoded in a polynomial number of bits, which is the case
in all reasonable implementations, the problem remains in PSPACE.

In general, we may safely assume that any generalization of a commercial
single-player game is in PSPACE. Indeed, these games are inherently designed
to run smoothly in real time, so they must rely on a game engine that is at
least in PSPACE, if not in P. Again, once a PSPACE game engine is given,
the game can be played nondeterministically, maintaining a polynomial-size
game state and using that NPSPACE = PSPACE.

\paragraph{Game model and glitches.}
We adopt an idealized model of the games in which we assume that
the rules of the games are as (we imagine) the game developers intended
rather than as they are implemented. In particular, we assume the absence
of major game-breaking glitches (for an example of a major game-breaking
glitch, see~\cite{gamebreak}, in which the speed runner ``beats'' Super Mario
World in less than 3 minutes by performing a sequence of seemingly arbitrary
and nonsensical actions, which fools the game into thinking the game is won).
We view these glitches not as inherently part of the game but rather as
artifacts of imperfect implementation.

\paragraph{Organization.}
\iffull
In Section~\ref{s:platform_framework}, we present two general schematics
used in almost all of our NP-hardness and PSPACE-hardness reductions.
We show that, if the basic gadgets in one of the two constructions
are implemented correctly, then the reduction from either 3-SAT or TQBF follows. In
Section~\ref{s:mario}, we prove that generalized Super Mario Bros.\ is
NP-hard by constructing the appropriate gadgets for the construction
given in Section~\ref{s:platform_framework}. In Sections~\ref{s:dkc},
\ref{s:zelda},~\ref{s:metroid}, and~\ref{s:pokemon}, we do the same for generalized
Donkey Kong Country, Legend of Zelda, Metroid, and Pok\'emon, respectively.
In addition, Sections~\ref{s:dkc} and~\ref{s:zelda} show that
the Donkey Kong Country games and some Legend of Zelda games are
PSPACE-complete, again by constructing the appropriate gadgets introduced in Section~\ref{s:platform_framework}.
\fi

\section{Frameworks for Platform Games}\label{s:platform_framework}
The general decision problem we consider for platform games is
to determine whether it is possible to get from a given start
point to a given goal point.  This is a natural problem because the main
challenge of platform games is to maneuver around enemies and
obstacles in order to reach the end of each stage.

\subsection{Framework for NP-hardness}

We use a general framework for proving the NP-hardness of platform
games, illustrated in Figure~\ref{fig:platform_framework}.
This framework is similar to the NP-hardness proof of PushPush-1 \cite{Push}.
With this framework in hand, we can prove hardness of individual games
by just constructing the necessary gadgets.

\begin{figure}[htbp]
\centering
\includegraphics[width=\linewidth]{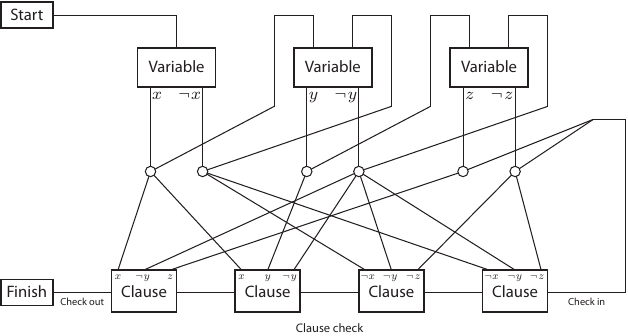}
\caption{General framework for NP-hardness}
\label{fig:platform_framework}
\end{figure}

The framework reduces from the classic NP-complete problem 3-SAT:
decide whether a 3-CNF Boolean formula can be made ``true'' by
setting the variables appropriately.
The player's character starts at the position labeled Start, then
proceeds to the Variable gadgets. Each Variable gadget forces the player
to make an exclusive choice of ``true'' ($x$) or ``false'' ($\neg x$)
value for a variable in the formula.
Either choice enables the player to follow paths leading to Clause gadgets,
corresponding to the clauses containing that literal ($x$~or~$\neg x$).
These paths may cross each other, but Crossover gadgets prevent the player
from switching between crossing paths.
By visiting a Clause gadget, the player can ``unlock'' the clause
(a permanent state change),
but cannot reach any of the other paths connecting to the Clause gadget.
Finally, after traversing through all the Variable gadgets,
the player must traverse a long ``check'' path, which passes through
each Clause gadget, to reach the Finish position.
The player can get through the check path if and only if each clause
has been unlocked by some literal.
Therefore, it suffices to implement Start, Variable, Clause, Finish,
and Crossover gadgets to prove NP-hardness of each platform game.

The specific properties our gadgets must satisfy are the following.

\paragraph{Start and Finish.} The Start and Finish gadgets contain the
starting point and goal for the character, respectively.  In most of our
reductions, these gadgets are trivial, but
in some cases we need the player to be in a certain state throughout
the construction, which we can force by making the Finish accessible only
if the player is in the desired state, and the desired state may be entered
at the Start. For example, in the case of Super Mario Bros., we need Mario
to be big throughout the stage, so we put a Super Mushroom at the start
and a brick at the Finish, which can be broken through only if Mario is big.

\paragraph{Variable.}
Each Variable gadget must force the player to choose one of two paths,
corresponding to the variable $x_i$ or its negation $\neg x_i$ being chosen
as the satisfied literal,
such that once a path is taken, the other path can never be traversed.
Each Variable gadget must be accessible from and only from
the previous Variable gadget, independent of the choice made in the
previous gadget, in such a way that entering from one literal
does not allow traversal back into the negation of the literal.

\paragraph{Clause and Check.} Each Clause gadget must be accessible from 
(and initially, only from) the literal paths
corresponding to the literals appearing in the clause in the original
Boolean formula.  In addition, when the player visits a Clause gadget in this
way, they may perform some action that ``unlocks'' the gadget.
The Check path traverses every Clause gadget in sequence,
and the player may pass through each Clause gadget via
the Check path if and only if the Clause gadget is unlocked.
Thus the Check path can be fully traversed only if all the Variable gadgets
have been visited from literal paths.
If the player traverses the entire Check path, they may access the Finish gadget.

\paragraph{Crossover.} The Crossover gadget must allow traversal via two passages
that cross each other, in such a way that there is no leakage between the two.

\begin{remark}\label{r:xover}
The Crossover gadget only needs to be \emph{unidirectional},
in the sense that each of the two crossing paths needs to be traversed
in only one direction. This is sufficient because, for each path
visiting a clause from a literal, instead of backtracking to the
literal after visiting the clause, we can reroute directly to visit
the next clause, so the player is never required to traverse a
literal path in both directions.
\end{remark}

\begin{remark}\label{r:xover2}
It is safe to further assume in a Crossover gadget that each of the two
crossing paths is traversed at most once,
and that one path is never traversed before the other path
(i.e., if both paths are traversed, the order of traversal is fixed).
This is sufficient because two literal paths either are the two sides of the
same Variable (and hence only one gets traversed),
or they come from different Variables,
in which case the one from the earlier Variable in the
sequence is guaranteed to be traversed before the other
(if it gets traversed at all).
Thus it is safe to have a Crossover gadget, featuring two crossing paths $A$ and~$B$, which after traversing path $B$ allows leakage from $A$ to~$B$.
(However, leakage from $B$ to $A$ must still be prevented.)
\end{remark}

\subsection{Framework for PSPACE-hardness}

For the PSPACE-hardness of Donkey Kong Country and Zelda: A Link to the Past, we apply a modified version of a framework described in~\cite[Metatheorem~4.c]{Games} and~\cite{Lemmings2}.

Our framework reduces from the PSPACE-complete problem True Quantified Boolean Formula (TQBF), and assumes the existence of a 
\emph{Door gadget} and a \emph{Crossover gadget}. A Door gadget's state may be \emph{open} or \emph{closed}, and the gadget may be ``traversed'' by the player if and only if it is open. A Door gadget also incorporates mechanisms to open and close it, which can be ``operated'' by the player.

\paragraph{Door gadget.} A Door gadget scheme is illustrated in Figure~\ref{fig:platform_framework_pspace_door}.

\begin{figure}[htbp]
\centering
\includegraphics[scale=1.75]{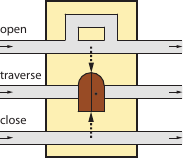}
\caption{Door gadget}
\label{fig:platform_framework_pspace_door}
\end{figure}

Three distinct paths enter the gadget from the left and exit to the right, without leakage. The ``traverse'' path implements the actual door, and may be traversed if and only if the gadget is in the open state. The other two paths allow the player to open and close the door, respectively. There is a subtle difference between the two: whenever the player traverses the ``close'' path, they trigger a mechanism that closes the door in the traverse path; when the player traverses the ``open'' path, they are \emph{allowed} to open the door, but they may choose not to.

\begin{remark}\label{r:door}
It will be evident that opening a door whenever possible is never a ``wrong'' move in the player's perspective. Hence, allowing the player to \emph{decide} whether to open a door or not upon traversal of the ``open'' path, rather than \emph{forcing} them to open it, has no real impact on the gadget's functioning, but often allows for less complicated implementations of the gadget itself.
\end{remark}

\begin{remark}
In the PSPACE-hardness framework of~\cite[Metatheorem~4.c]{Games}, the approach is slightly different, in that Door gadgets do not incorporate any mechanism to operate them, but are operated from a distance, by walking on \emph{pressure plates}. The advantage of such an approach is that it requires no Crossover gadget, but the drawback is that implementing pressure plates acting on arbitrarily distant doors may be problematic in some games.
\end{remark}

\paragraph{Framework.} The general framework yielding our reduction is sketched in Figure~\ref{fig:platform_framework_pspace}.

\begin{figure}[htbp]
\centering
\includegraphics[width=\linewidth]{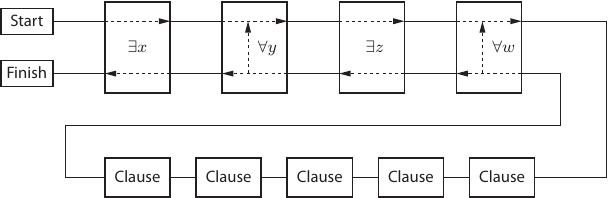}
\caption{General framework for PSPACE-hardness}
\label{fig:platform_framework_pspace}
\end{figure}

A given fully quantified Boolean formula $\exists x \forall y \exists z \cdots \varphi(x, y, z, \cdots)$, where $\varphi$ is in 3-CNF, is translated into a row of \emph{Quantifier gadgets}, followed by a row of \emph{Clause gadgets}, connected by several paths.

Traversing a Quantifier gadget at any time sets the truth value of the corresponding Boolean variable. On the other hand, each Clause gadget can be traversed if and only if the corresponding clause of $\varphi$ is satisfied by the current variable assignments. Whenever traversing an Existential Quantifier gadget, the player can choose the truth value of the corresponding variable. On the other hand, the first time a Universal Quantifier gadget is traversed, the corresponding variable is set to true.

When all the variables are set, the player attempts to traverse the Clause gadgets. If the player succeeds, they proceed to the ``lower parts'' of the Quantifier gadgets, where they are rerouted to the last Universal Quantifier gadget in the sequence. The corresponding variable is then set to false, and $\varphi$ is ``evaluated'' again by making the player walk through all the Clause gadgets.

The process continues, forcing the player to ``backtrack'' several times, setting all possible combinations of truth values for the universally quantified variables, and choosing appropriate values for the existentially quantified variables in the attempt to satisfy $\varphi$.

Finally, when all the necessary variable assignments have been tested and $\varphi$ keeps being satisfied, i.e., if the overall quantified Boolean formula is true, the exit becomes accessible, and the player may finish the level. Conversely, if the quantified Boolean formula is false, there is no way for the player to operate doors in order to reach the exit.

Next we show how to implement all the components of our framework using just Door gadgets, paths, and Crossover gadgets.

\paragraph{Clause.} Clause gadgets are implemented as sketched in Figure~\ref{fig:platform_framework_pspace_clause}.

\begin{figure}[htbp]
\centering
\includegraphics[scale=2.75]{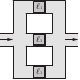}
\caption{Clause gadget. Labeled squares represent Door gadgets.}
\label{fig:platform_framework_pspace_clause}
\end{figure}

There is a Door gadget for each literal in the clause, and the player may traverse the clause if and only if at least one of the doors is open. (For simplicity, the ``open'' and ``close'' paths of the Door gadgets are not displayed.)

\paragraph{Existential Quantifier.} In Figure~\ref{fig:platform_framework_pspace_existential}, the Existential Quantifier gadget for variable $x$ is represented. $x_1$, $x_2$, etc.\ (resp., $\bar x_1$, $\bar x_2$, etc.) denote the positive  (resp., negative) occurrences of $x$ in the clauses of $\varphi$.

\begin{figure}[htbp]
\centering
\includegraphics[scale=2.75]{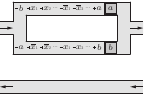}
\caption{Existential Quantifier gadget. ``$+a$'' (resp., ``$-a$'') is shorthand for ``open (resp., close) the door labeled $a$''.}
\label{fig:platform_framework_pspace_existential}
\end{figure}

When traversing the upper part of the gadget from left to right, the player must choose one of the two paths, thus setting the truth value of $x$ to either true or false. This is done by appropriately opening or closing all the doors corresponding to occurrences of $x$ in $\varphi$. For instance, in order to set the first occurrence of literal $x$ to false, the player is diverted to the ``close'' path of the Door gadget labeled $x_1$, and then rerouted to the Existential Quantifier gadget (this is formally equivalent to walking on a pressure plate, in the terminology of~\cite{Games}). If the paths to and from the Door gadgets intersect some other paths, then Crossover gadgets are used, as Figure~\ref{fig:platform_framework_pspace_legend} exemplifies.

\begin{figure}[htbp]
\centering
\includegraphics[width=\linewidth]{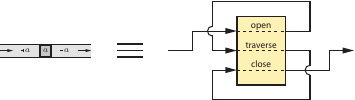}
\caption{Implementing Quantifier gadgets with Door gadgets. The item on the left is implemented as shown on the right.}
\label{fig:platform_framework_pspace_legend}
\end{figure}

The doors labeled $a$ and $b$ in Figure~\ref{fig:platform_framework_pspace_existential} prevent leakage between the two different paths of the Existential Quantifier gadget, enforcing mutual exclusion.

Finally, the lower part of the gadget is traversed from right to left when the player backtracks, and it is simply a straight path.

\paragraph{Universal Quantifier.} A Universal Quantifier gadget for variable $x$ is shown in Figure~\ref{fig:platform_framework_pspace_universal}.

\begin{figure}[htbp]
\centering
\includegraphics[scale=2.75]{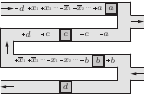}
\caption{Universal Quantifier gadget}
\label{fig:platform_framework_pspace_universal}
\end{figure}

When the player enters the gadget from the top left, door $d$ gets closed and variable $x$ is set to true. Then the player must exit to the top right, because door $c$ cannot be traversed from right to left.

When backtracking the first time, the player enters from the bottom right and, because door $d$ is still closed, they must take the upper path, thus setting variable $x$ to false. Incidentally, door $d$ gets opened and door $a$ gets closed, thus preventing leakage to the top left entrance, and forcing the player to exit to the top right again.

When backtracking the second time (i.e., when both truth values of $x$ have been tested), door $d$ is open and the player may finally exit to the bottom left. When done backtracking, the player will eventually enter this gadget again from the top left, setting $x$ to true again, etc.

\paragraph{Remarks.} We observe that, as a result of our constructions, each Door gadget is traversed by exactly three disjoint paths, in accordance with Figure~\ref{fig:platform_framework_pspace_door}. For instance, the door labeled $x_1$ is traversed by one path in the Clause gadget where it belongs, and is opened (resp., closed) via exactly one path going to and from the Quantifier gadget corresponding to variable $x$. Similarly, every other door in each Quantifier gadget is reached by exactly one ``open'' path and one ``close'' path.

As noted in Remark~\ref{r:door}, opening a door has the only possible effect of making new areas accessible to the player. Therefore, we may safely assume that the player always chooses to open a door when traversing its ``open'' path. Indeed, if the player chooses not to open some door at some point, and still manages to reach the level exit, then \emph{a fortiori} they can reach the exit by opening the door.

Finally, we observe that it does not matter if the three paths of a Door gadget are traversable in just one way, or in both ways. Moreover, the opening mechanism can be simplified, and implemented as a single ``dead end path'' reaching the gadget from the left, which is traversed in both ways. This contrasts with the open mechanism sketched in Figure~\ref{fig:platform_framework_pspace_door}, which enters and exits at opposite sides of the gadget. We will take advantage of this aspect when implementing the Door gadgets for Donkey Kong Country~1 and~2 (see Figures~\ref{PSPACEDKC1} and~\ref{PSPACEDKC2}).

\section{Super Mario Bros.}\label{s:mario}

Super Mario Bros.\ is a platform video game created by Shigeru Miyamoto,
developed by Nintendo, published for
the Nintendo Entertainment System (NES) in 1985. In the game, the player
controls Mario, an Italian plumber, and must navigate through worlds full of
enemies and obstacles to rescue the kidnapped Princess Toadstool from
Bowser. The game  spawned over a dozen sequels
(SMB~2, SMB~3, Super Mario World, Super Mario 64, Super Mario Sunshine,
Super Mario Galaxy, etc.)\
across various gaming platforms.
Mario's character became so popular that he became Nintendo's mascot
and appeared in over 200 games.

The gameplay in Super Mario Bros.\ is simple.  
Initially, Mario is small (one tile tall) and dies upon harmful contact with an enemy. 
If Mario manages to obtain a Super Mushroom, then he becomes Super Mario,
who is two tiles tall. Upon any harmful contact, Super Mario reverts back to normal Mario.
If Mario obtains a Star, then he becomes invincible for a limited time.
Mario's basic actions are limited to walking, running,  crouching, and jumping. The height of a jump is normally 4 tiles, and it becomes 5 tiles when Mario runs.  Mario can perform more advanced actions by combining basic moves. For
example, crouching while running results in a {\em crouch-slide}, which will allow him
to fit into narrow corridors while as Super Mario.  Mario can also jump while
crouching, allowing him to jump through narrow gaps.

Apart from solid ground, the basic environment also contains walls and
mid-air platforms. These are typically constructed from three types of
blocks: normal blocks, which are inert; item blocks, which release an
item when hit by Mario from below; and bricks, which are destroyed upon
activation. Mario usually obtains items by hitting item blocks.

Jumping on top of most enemies is non-harmful, specifically for all those used
in our proof, though there are exceptions not
encountered in our proof. In fact, Mario primarily defeats enemies by jumping on them.    
A staple enemy of the Mario series is the Goomba, which is a corrupted mushroom
that walks back and forth and dies upon being stomped by Mario. Another key enemy
is the Koopa, which does not play a role in our construction for Super Mario Bros.\
but will appear in one of our gadgets for Super Mario Bros.\ 3. Koopas come in two
varieties. Both will keep walking in a direction, turning around if a wall is encountered,
but Green Koopas fall off cliffs whereas Red Koopas treat cliffs as walls. We use only
Red Koopas in our construction. When Mario jumps on a Koopa, it is not instantly
defeated, but instead it temporarily hides in its shell, which Mario can then kick.
Once kicked, a Koopa shell slides with constant velocity and bounces off surfaces,
destroying bricks and activating item blocks if it hits them from the side.
Finally, we will use Firebars, which are harmful segments of fire that rotate about an
endpoint.
There are many other items and enemies in the game, but they do not play a role
in our proofs so we do not include them here. For more information on the game,
see~\cite{mariowiki}.

In this section, we prove the following:
\begin{theorem}\label{t:mario}
It is NP-hard to decide whether the goal is reachable from the start
of a stage in generalized Super Mario Bros.
\end{theorem}
\begin{proof}
When generalizing the original Super Mario Bros., we assume that the
screen size covers the entire level, because the game forbids Mario from
going left of the screen. This generalization is not needed in later games,
because those games allow Mario to go left.

We use the general framework provided in
Section~\ref{s:platform_framework}, so it remains only to implement
the gadgets. The Start and Finish gadgets are straightforward. The
Start gadget, shown in Figure~\ref{fig:mario_start},
includes an item block containing a Super Mushroom which
makes Mario into Super Mario.
The Super Mushroom serves two purposes. First, Super Mario is two tiles
tall, which prevents him from fitting into narrow horizontal
corridors, a property essential to our other gadgets. Second, Super
Mario is able to destroy bricks whereas normal Mario cannot. In order
to force the player to take the Super Mushroom in the beginning,
we block Mario's path with bricks in
the Finish gadget, shown in Figure~\ref{fig:mario_finish}.

\ifabstract
\else

\begin{figure}[htbp]
\centering
\includegraphics[scale=0.5]{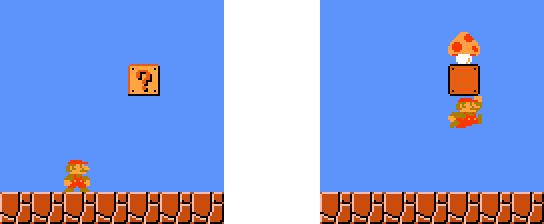}
\caption{Left: Start gadget for Super Mario Bros.
Right: The item block contains a Super Mushroom}
\label{fig:mario_start}
\end{figure}

\begin{figure}[htbp]
\centering
\includegraphics[scale=0.5]{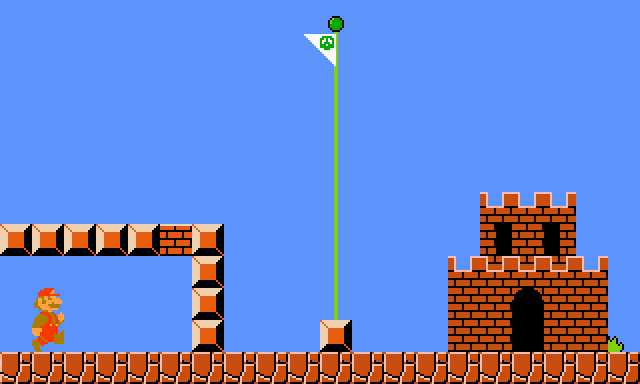}
\caption{Finish gadget for Super Mario Bros.}
\label{fig:mario_finish}
\end{figure}

\fi

Next, we implement the Variable gadget, illustrated in
Figure~\ref{fig:mario_var}.  There are two entrances, one from each
literal of the previous variable (if the variable is $x_i$, the two
entrances come from $x_{i-1}$ and $\neg x_{i-1}$).  Once Mario
falls down, he cannot jump back onto the ledges at the top, so Mario
cannot go back to a previous variable. In particular, Mario cannot go
back to the negation of the literal he chose. To choose which value to
assign to the variable, Mario may fall down either the left passage or
the right.

\iffull
\begin{figure}[htbp]
\centering
\includegraphics[scale=0.5]{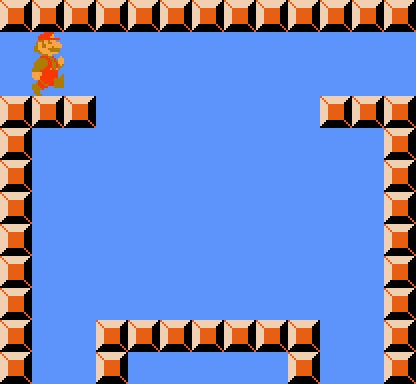}
\caption{Variable gadget for Super Mario Bros.}
\label{fig:mario_var}
\end{figure}
\fi

Now we present the Clause gadget, illustrated in
Figure~\ref{fig:mario_clause}.  The three entrances at the bottom
correspond to the three literals that are in the clause. When the Clause
is visited, Mario hits the item block, which releases a Star into
the area above. The Star will bounce around that area until Mario comes through the Check
path to pick it up, allowing him to traverse through the Firebars safely.
Without a Star, Mario cannot get through the Firebars without getting hurt. Note that, despite the fact that a Star's invulnerability effects on Mario last just a few seconds, the Star itself never vanishes until it is collected by Mario.

\iffull
\begin{figure}[htbp]
\centering
\includegraphics[scale=0.5]{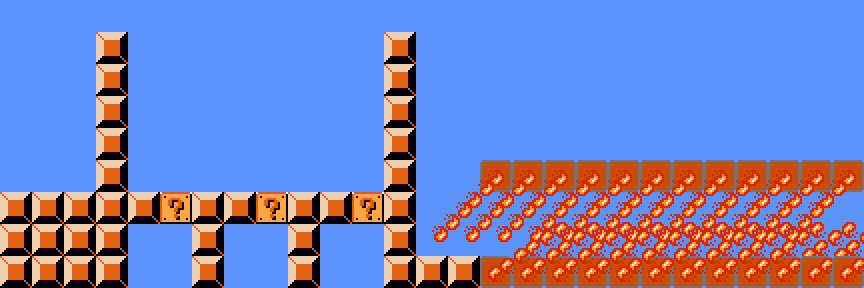}
\caption{Clause gadget for Super Mario Bros.}
\label{fig:mario_clause}
\end{figure}
\fi

Finally, we implement the Crossover gadget, illustrated in
Figure~\ref{fig:mario_xover_smb}.  This unidirectional gadget
enables Mario to traverse left-to-right or bottom-to-top, which suffices
by Remark~\ref{r:xover}.
Mario cannot enter this Crossover initially from the top or right.

Approaching from the left, Mario may
get hurt by the Goomba and revert back to his small state, then walk
through the narrow corridor and pick up a Super Mushroom from
the item block on the other side, allowing him to break through the brick.
Initially, there is no leakage to the bottom nor leakage to the top because
small Mario cannot break bricks. Furthermore, there is not enough space for
Mario to run and crouch-slide through the gap.

Approaching from the bottom,
Mario may break the right brick on the lower level, jump onto the left brick
of the lower level, break the left brick on the upper level, and jump onto the
right brick of the upper level to proceed upwards. Mario must be big
while traversing vertically in order to break the bricks, and because there is
a one-way passage leading to the bottom entrance, Super Mario cannot
break the bricks first, then return later as small Mario. Because Mario must be big
this entire time, he cannot fit through the narrow gaps on the sides and
so there is no horizontal leakage.

There may be leakage from the horizontal path to the vertical path \emph{after} traversing the vertical path. This is because, once the bricks have been broken, even small Mario can exit from the top. Fortunately, Remark~\ref{r:xover2} says that such a leakage can safely be allowed. It suffices that no leakage can occur from the vertical path to the horizontal one.
\end{proof}

\iffull
\begin{figure}[htbp]
\centering
\includegraphics[scale=0.5]{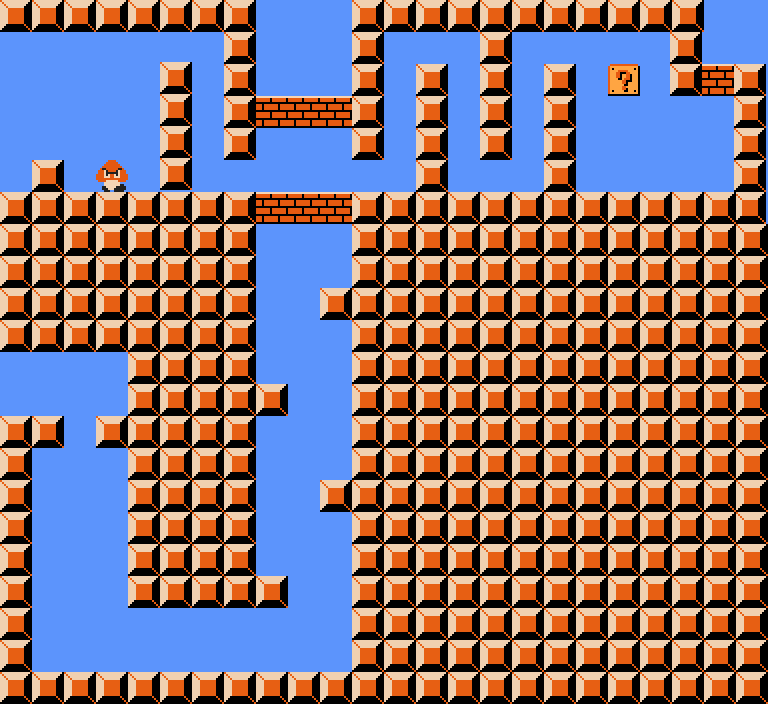}
\caption{Crossover gadget for Super Mario Bros.}
\label{fig:mario_xover_smb}
\end{figure}
\fi

\later{
\ifabstract
\else
  \paragraph{Other Mario games.}
\fi
The NP-hardness from Theorem~\ref{t:mario} holds just
as well for the NES sequel, SMB: The Lost Levels (known as
SMB~2 in Japan) because the gameplay mechanics are exactly the same.
However, to extend the result to later games, we need to modify
the Clause gadget (for both Super Mario Bros.\ 3 and Super Mario World) and
the Crossover and Finish gadgets (for SMW). Because Firebars appear in neither SMB~3 nor SMW,
we have to replace them with something else in the Clause gadget. For SMB~3, we replace
the Firebars with a vertical tunnel of bricks, as shown in Figure~\ref{fig:mario_clause_smb3}.
The three entrances corresponding to the literals are now on top,
with a Red Koopa in each entrance. To unlock the gadget, Mario
jumps on the Koopa and kicks its shell down, which
breaks open the brick columns and then falls down into the pit, allowing
Mario to pass through but not allowing him to re-use the shell for other clauses.
(This gadget cannot be used for Super Mario Bros.\ because Koopa shells do
not break blocks in that game.)

\begin{figure}[htbp]
\centering
\includegraphics[scale=0.5]{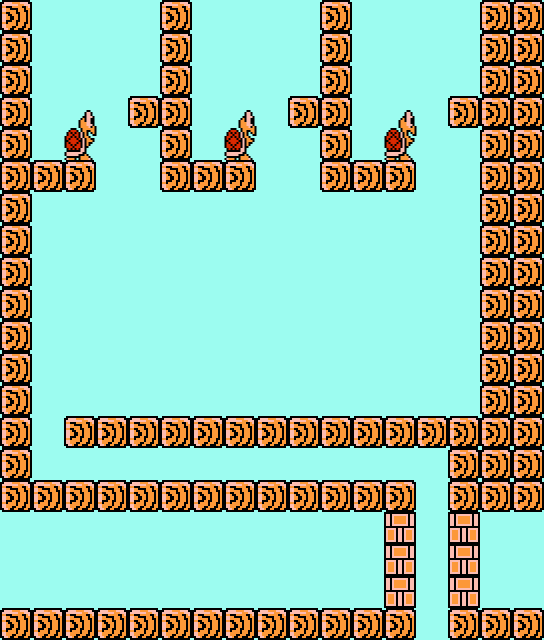}
\caption{Clause gadget for Super Mario Bros.\ 3}
\label{fig:mario_clause_smb3}
\end{figure}

For SMW, we modify the Clause gadget for SMB by replacing the Firebars
with Munchers, which are indestructible black plants attached to a ground
surface which harms Mario when walked on.

To modify the Crossover gadget for SMW, we replace bricks with rotating blocks, which
Super Mario may destroy by spin jumping on them from above; see
Figure~\ref{fig:mario_xover_smw}.
A \emph{spin jump} is a powerful jump performed by Mario, which destroys rotating blocks
he lands on, but only if he is big.
This crossover is again unidirectional, now supporting
left-to-right and top-to-bottom traversals. Moreover, once again there may be leakage from the horizontal passage to the vertical passage after the latter has been traversed, but there is no leakage from vertical to horizontal. As before, this suffices by Remark~\ref{r:xover2}.

\begin{figure}[htbp]
\centering
\includegraphics[width=\linewidth]{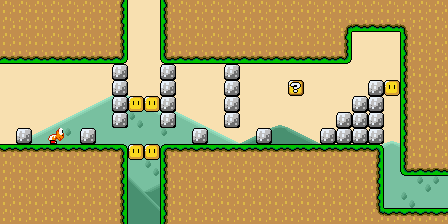}
\caption{Crossover gadget for Super Mario World}
\label{fig:mario_xover_smw}
\end{figure}

The Finish gadget for SMW is a straightforward modification, shown in
Figure~\ref{fig:mario_finish_smw}.
\begin{figure}[htbp]
\centering
\includegraphics{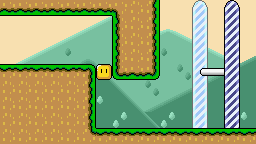}
\caption{Finish gadget for Super Mario World}
\label{fig:mario_finish_smw}
\end{figure}

American SMB~2 (known as Super Mario USA in Japan) has somewhat different
mechanics from those of the other SMB games.
The main new mechanic is Mario's ability to uproot plants growing in the
ground, and also to stand on many types of enemy and pick them up.
After doing so, Mario can walk and jump while carrying one such object overhead,
and finally he can throw the object.

In American SMB~2, Mario can jump 5 tiles high by running or crouch-jumping
(this is done by crouching for a short amount of time to ``charge'' and then
jumping while charged).
However, he can jump 3 tiles higher if he stands on a Ninji
(who jumps up and down cyclically) and waiting for the appropriate time.
In the following gadgets, we assume Mario to be small.
He normally starts big, but we can enforce smallness by forcing him to fall
into a pit full of spikes as the level starts.
In order to become big again, he would
have to find mushrooms (of which there are none in our levels) or kill a
certain number of enemies (but we will ensure that he can kill no enemy).

\begin{figure}[htbp]
\centering
\includegraphics{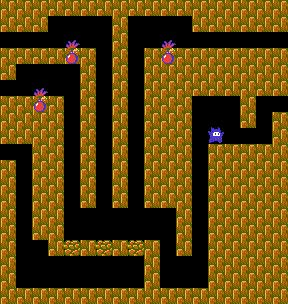}
\caption{Clause gadget for American SMB 2}
\label{fig:mario_clause_smb2}
\end{figure}

Mario visits the Clause gadget (Figure~\ref{fig:mario_clause_smb2})
from one of the three routes containing a bomb plant.
Mario can pick up the bombs and throw them to the breakable tiles below.
Then he traverses the gadget from right to left.
He has to pick up the Ninji and drop it in the bottom-left corner of the gadget
in order to jump 8 tiles and exit to the left. Note that if he falls down 
after throwing a bomb, he gets stuck because he has no Ninji.
Also, the Ninji is insufficient to let him jump from the lower part of the
gadget to the literal paths.

\begin{figure}[htbp]
\centering
\includegraphics{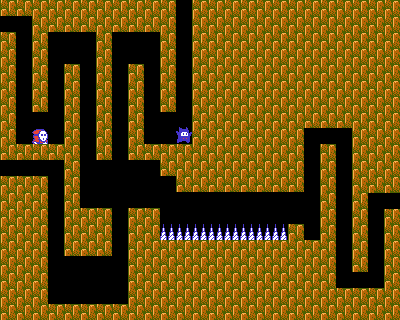}
\caption{Crossover gadget for American SMB 2}
\label{fig:mario_xover_smb2}
\end{figure}

In the Crossover gadget (Figure~\ref{fig:mario_xover_smb2}), Mario is allowed
to go from left to right first, and from top to bottom after.
To traverse left to right, he has to pick up the red Shy Guy, throw it onto
the spikes, and jump on its head. Then the Shy Guy will lead Mario safely to
the right.
Because the Shy Guy is red (and not pink), it will not turn around at the edge
of the platform after the spikes, and instead will fall into the pit, and
allow Mario to jump to the upper platform and exit.
If Mario tries to put the Shy Guy anywhere else, he either dies or gets stuck.
To traverse top to bottom, Mario picks up the Ninji, takes it to the bottom
platform, and uses it to make a high jump to the left.
The Ninji is insufficient to jump to the Shy Guy area, and the Shy Guy is
insufficient to allow Mario to perform a high jump.
On the other hand, the Ninji cannot be used to traverse the spikes, because
it only jumps in place.

In order to kill enemies (and get hearts that make him grow), Mario would
have to throw something (either a bomb or another enemy) at them.
However, by our design, Mario cannot carry enemies from one gadget to another,
and he can never reach more than one object (bomb or enemy) at a time within
a gadget. Specifically, in order to make sure that Mario cannot carry an enemy from one gadget to another, we can connect gadgets with 1-tile-high tunnels that also include jumps and long falls, as the right path in Figure~\ref{fig:mario_xover_smb2} suggests. Since Mario cannot carry an enemy through a 1-tile-high tunnel (the tunnel has to be at least 2 tiles high for him to carry an enemy overhead) he would have to throw such an enemy into the tunnel, and then follow it. This means in particular that the enemy must be a Shy Guy, because a Ninji does not walk. But if the tunnel ends with a wall that has to be jumped, there is no way Mario can pick up the Shy Guy again and carry it above the wall without being killed. Also note that, in a Crossover gadget, Mario cannot use the Ninji to kill the Shy Guy that may be in the pit. Indeed, Mario always throws objects diagonally downwards at a fixed angle, and when a Ninji that has been thrown hits the ground, its momentum makes it bounce forward by three or four tiles only.

Concluding, we have the following.
\begin{theorem}\label{t:mario2}
It is NP-hard to decide whether the goal is reachable from the start
of a stage in generalized Super Mario Bros.~1--3, The Lost Levels, and Super Mario World.\hfill\qed
\end{theorem}

}
\section{Donkey Kong Country}\label{s:dkc}

Donkey Kong Country is a platform video game developed by Rare Ltd.\ and
published for the Super Nintendo Entertainment System (SNES) in 1994.
Like Mario, the character of Donkey Kong was originally created by
Shigeru Miyamoto, and appears in many other games as well.  The
premise of Donkey Kong Country is isomorphic to that of Super Mario
Bros.  In particular, King K.\ Rool has stolen Donkey Kong's banana
hoard, and the player controls Donkey Kong (and occasionally his small
nephew Diddy Kong) to navigate through levels filled with enemies and
obstacles in order to defeat K.\ Rool.

The gameplay of Donkey Kong Country is also quite simple. We only
discuss gameplay mechanics necessary for our proof. For more
information on the game, see~\cite{dkwiki}. Donkey Kong's basic
actions are walking, running, rolling, crouching, and jumping. He may
also pick up Barrels and throw them, after which they roll along the
ground until they hit a wall or an enemy, killing it. Donkey Kong can
jump into Barrel Cannons, from which he may be fired in a straight
line, defying gravity. Manual Barrel Cannons cycle through a preset
set of angles at which Donkey Kong may be fired, which the player may
choose with proper timing. Automatic Barrel Cannons fire Donkey Kong
as soon as he jumps in. Zingers are enemies that resemble giant
bees. They kill Donkey Kong on contact, and cannot be killed except by Barrels.


\later{
\ifabstract
\else
  \subsection{NP-hardness}
\fi
\label{donkey NP}

In this section, we prove the following:
\begin{theorem}\label{t:dkc}
It is NP-hard to decide whether the goal is reachable from the start
of a stage in generalized Donkey Kong Country.
\end{theorem}
\begin{proof}
We show that Donkey Kong Country is NP-hard by reducing from 3-SAT. As
in the proof for Super Mario Bros., we use the general framework
provided in Section~\ref{s:platform_framework}.
The Start and Finish gadgets are
trivial; they are simply the beginning and end parts of the level,
respectively. The Variable gadget is analogous to the Variable
gadget for Mario, replacing blocks with the appropriate terrain from
Donkey Kong Country. The Clause gadget, illustrated in
Figure~\ref{fig:dkc_clause}, is also similar to that in Mario. When
visiting the Clause gadget from a Variable gadget, Donkey Kong must pick up the
Barrel and toss it down, after which it rolls and kills the Zinger
(bee).  The Barrel Cannon above each platform is used to return
upward, and the platforms are sufficiently high that, if Donkey Kong
drops down, he cannot jump back up onto the platforms.

\begin{figure}[htbp]
\centering
\includegraphics[scale=0.5]{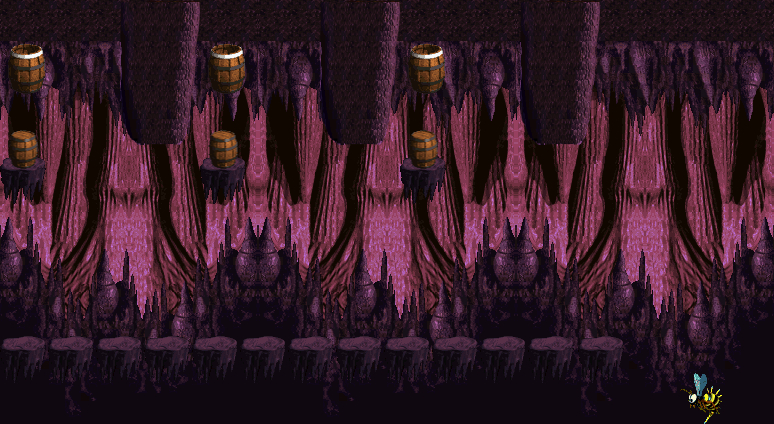}
\caption{Clause gadget for Donkey Kong}
\label{fig:dkc_clause}
\end{figure}

The final Check path, illustrated in Figure~\ref{fig:dkc_check}, is
straightforward.  The illustration is not drawn to scale, but the idea
is that at the end, Donkey Kong enters the Automatic Barrel Cannons on
the left, which blast him towards the Automatic Barrel Cannons on the
right, which blast him downward toward the goal. There are Zingers
positioned in the path, one per clause (these are the same Zingers
that are in the Clause gadgets) so that Donkey Kong can reach the goal
if and only if every Zinger has been killed.

\begin{figure}[htbp]
\centering
\includegraphics[width=\linewidth]{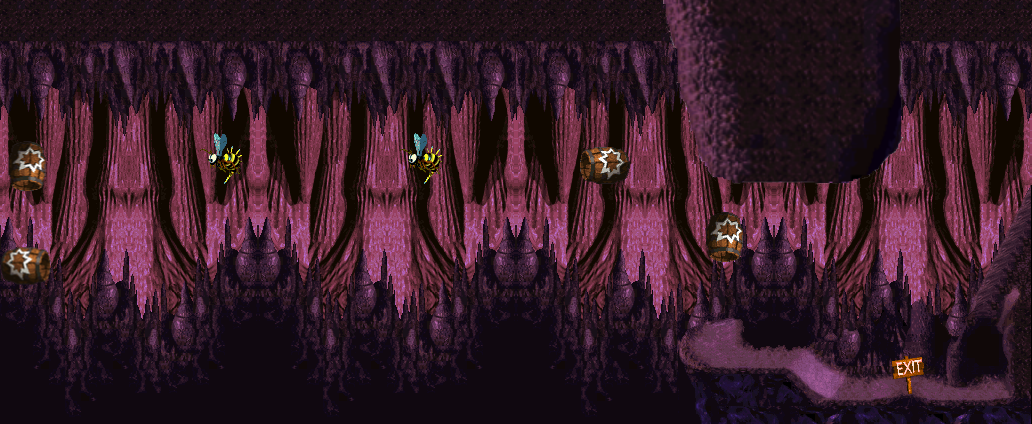}
\caption{Clause check for Donkey Kong}
\label{fig:dkc_check}
\end{figure}

Finally, we must implement the Crossover gadget (Figure~\ref{fig:dkc_xover}),
which is remarkably easy in Donkey Kong Country.
Essentially, at every entrance/exit, there is a forward Automatic Barrel Cannon
and a backward one. Therefore, the horizontal passage can be traversed
without leakage to the vertical passage, and vice versa.
\end{proof}

\begin{figure}[htbp]
\centering
\includegraphics[scale=0.5]{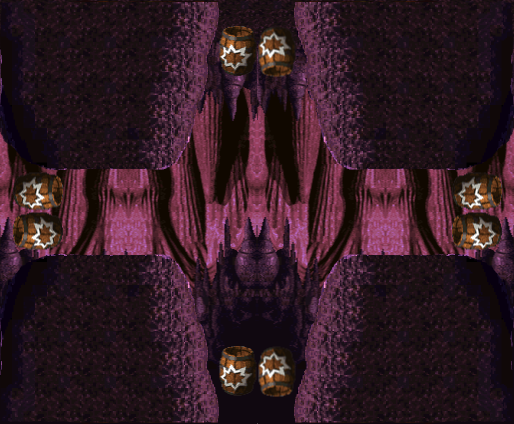}
\caption{Crossover gadget for Donkey Kong}
\label{fig:dkc_xover}
\end{figure}


\paragraph{Other Donkey Kong games.}
Theorem~\ref{t:dkc} holds just as well for the sequels---Donkey Kong Country~2 and~3---because Zingers, Barrels, and Barrel Cannons are all present in these games.
}
\subsection{PSPACE-completeness}

We can prove PSPACE-completeness of Donkey Kong Country and its sequels, using the general PSPACE-hardness framework given in Section~\ref{s:platform_framework}.

\begin{theorem} \label{DKC1 PSPACE}
It is PSPACE-complete to decide whether the goal is reachable from the start of a stage in generalized Donkey Kong Country\ 1.
\end{theorem}
\begin{proof}
In all Donkey Kong Country games, we may assume that the player controls only a single Kong, by placing a DK barrel (a barrel containing the backup Kong member) at the start of the level, followed by a wall of red Zingers (which are not killable by Barrels).

We use an object from the game, the Tire, which behaves as follows: it is an object which can be pushed
horizontally by the player, which causes it to roll in that direction, and if jumped on, the player bounces
up very high.

\ifabstract
\else
The Crossover gadget has already been described,
\fi
and the Door gadget is depicted in Figure~\ref{PSPACEDKC1}.
The door is closed if the Tire is located as shown in the picture, and is open if it is located up the slide. The ground is made of ice, so that both the Tire and the Kongs slide on it when they gain some speed. The right-facing Zingers (yellow) are static, while the left-facing ones (red) move from left to right in swarms, as indicated by arrows.

If the door is closed, the player attempting to traverse it from the top Barrel will land on the Tire and will be stuck because of the surrounding Zingers. If the door is open, the player lands on the ground, slides rightward below the Zingers while crouching, and safely reaches the pit on the right.

If the player comes from the ``open'' entrance, and the door is closed, they can go left, push the Tire up the slope, and safely come back. The player cannot jump down the pit, because the red Zingers are covering it.

If the player comes from the ``close'' path, they must climb on the rope (when the red Zingers move to the right) and reach the upper path. If the Tire is on the edge of the slope (i.e., the door is open), the player incidentally pushes it to the right, thus making it slide down, and hence closing the door.

Membership in PSPACE follows from the remarks in the Introduction.
\end{proof}

\begin{figure}[htbp]
\centering
\includegraphics[width=\linewidth]{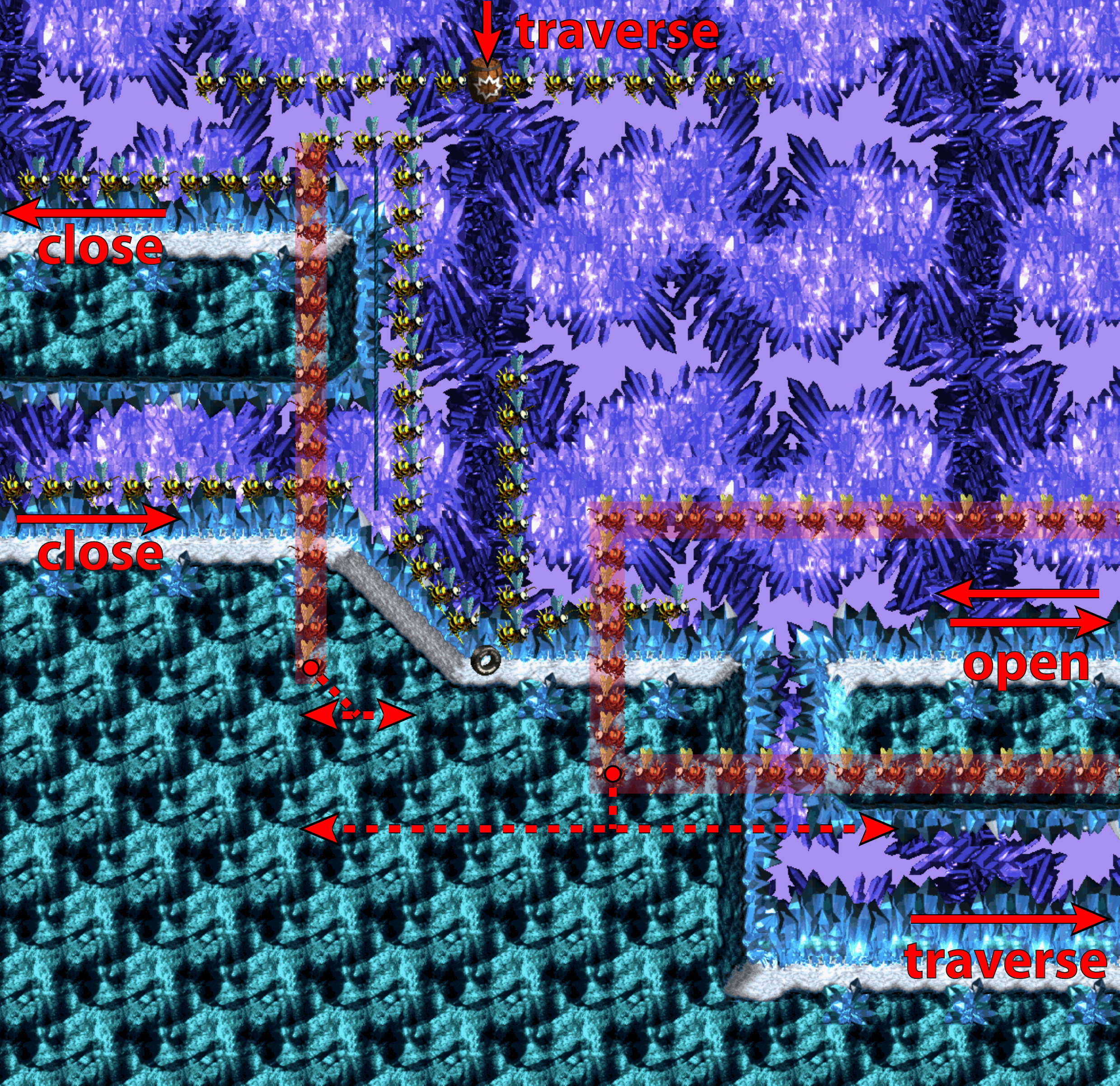}
\caption{Door gadget for Donkey Kong Country 1}
\label{PSPACEDKC1}
\end{figure}

\later{
\ifabstract
\else
  \paragraph{Donkey Kong Country sequels.}
\fi
\label{donkey appendix}

\begin{theorem}
It is PSPACE-complete to decide whether the goal is reachable from the start of a stage in generalized Donkey Kong Country~2.
\end{theorem}
\begin{proof}
We use an object from the game called a Hot Air Balloon, which works as follows: it acts as a platform
on which the player may stand, and by default it slowly descends, unless it is above a Hot Air Current
(a pillar of hot air), which forces it to ascend --- while in the air, the player may influence the horizontal
direction of the balloon.

The Door gadget is illustrated in Figure~\ref{PSPACEDKC2}, where the door is closed if the blue Hot Air Balloon is positioned as depicted, and is open if the Balloon lies in the position marked with letter A.

When the player attempts to traverse the door from the Barrel above, they are shot between the two swarms of red Zingers (each Zinger moves left and right in a sinusoidal fashion), until they safely reach the Barrel below, if the door is open. Otherwise, if the door is closed, the player will land on the Balloon and will be stuck among Zingers.

If the player comes from the ``open'' entrance, they can decide to open the door, if it is closed.  To do so, they climb on the Chain, wait for the red Zingers to move to the right, jump on the Balloon and quickly pull it to the left. Then the player jumps back to the Chain before hitting the red Zingers below, while the Balloon floats downward, across all the Zingers, and stops right above the lower left Barrel (letter A in Figure~\ref{PSPACEDKC2}).

If the player comes from the ``close'' entrance, and the door is already closed, they will be shot down (avoiding the Zingers, if the timing is right) and safely reach the ``close'' exit.  Otherwise, if the door is open, the player will land on the Balloon and will have to quickly maneuver it to the right to escape the
red Zingers, until the long Hot Air Current will send the Balloon up (before the player has a chance to drive it back to the left) and it will ultimately close the door, while the player exits to the right upon entering the Barrel.

Membership in PSPACE follows from the remarks in the Introduction.
\end{proof}

\begin{figure}[htbp]
\centering
\includegraphics[width=0.8\linewidth]{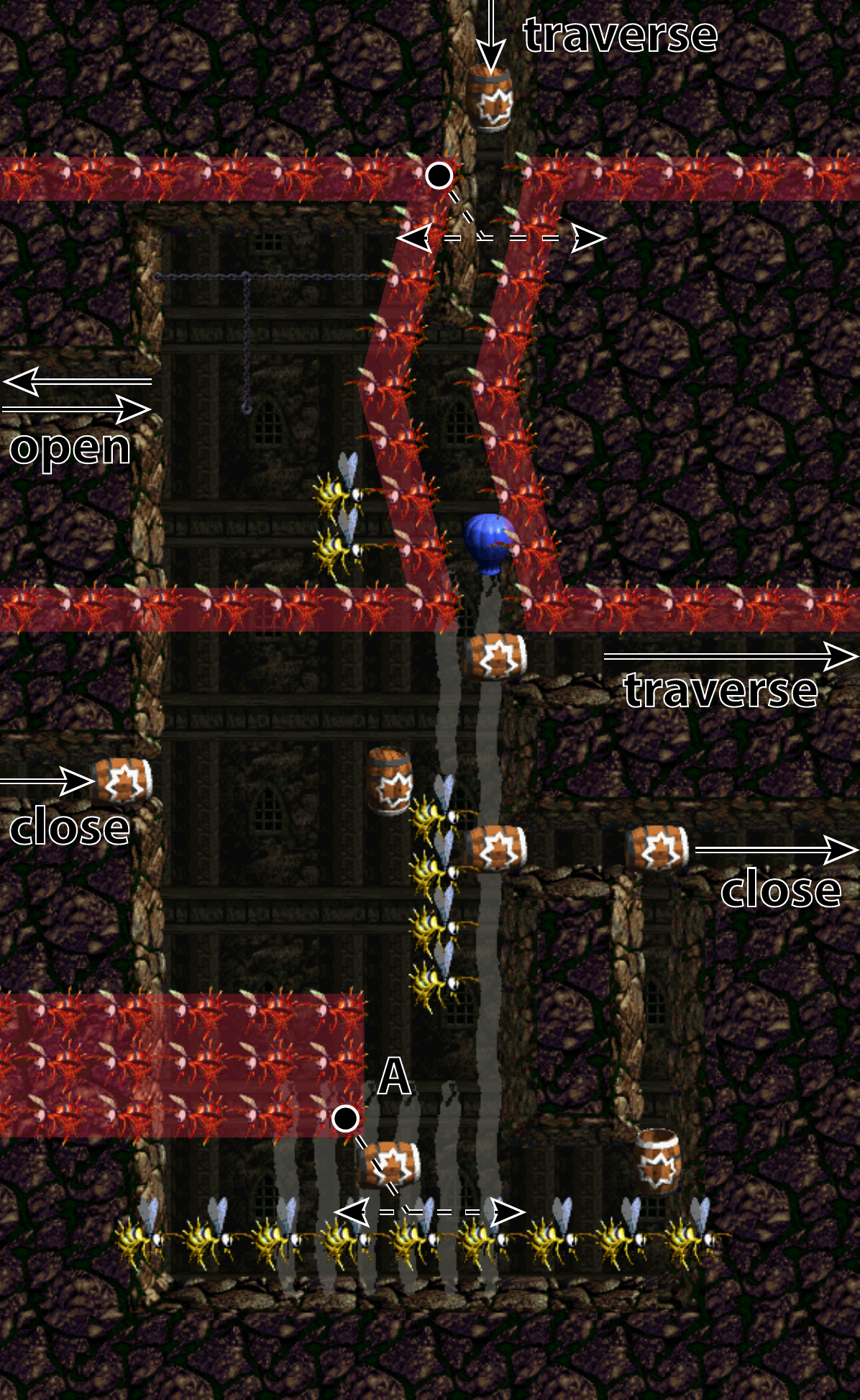}
\caption{Door gadget for Donkey Kong Country~2}
\label{PSPACEDKC2}
\end{figure}

\begin{theorem}
It is PSPACE-complete to decide whether the goal is reachable from the start of a stage in generalized Donkey Kong Country\ 3.
\end{theorem}
\begin{proof}
The Door gadget is shown in Figure~\ref{PSPACEDKC3}, where the door is open if and only if the Tracker Barrel is on the left, as in the picture. Operating each Handle opens the Gate with the same number, which then closes again after a few seconds.

Coming from the ``traverse'' entrance, if the player finds the Tracker Barrel below, they can use it to reach the ``traverse'' exit; otherwise they fall in the bottom corridor and they are stuck (jumping back up is impossible because it is too high).

Kongs cannot ``escape'' a Tracker Barrel after they enter it, because it keeps shooting them in the air, but follows their movements horizontally wherever they go, catching them again as they fall back down. The Tracker Barrel stops only if the Kongs enter a new Barrel after they are shot (which is irrelevant here) or if the Tracker Barrel itself reaches its leftmost or rightmost position. For this reason, coming from the ``close door'' entrance when the door is open forces the player to move the Tracker Barrel all the way to the right, hence closing the door. Likewise, coming from the ``open door'' entrance when the door is closed makes the player move the Tracker Barrel all the way to the left, opening the door.

Membership in PSPACE follows from the remarks in the Introduction.
\end{proof}

\begin{figure}[htbp]
\centering
\includegraphics[width=\linewidth]{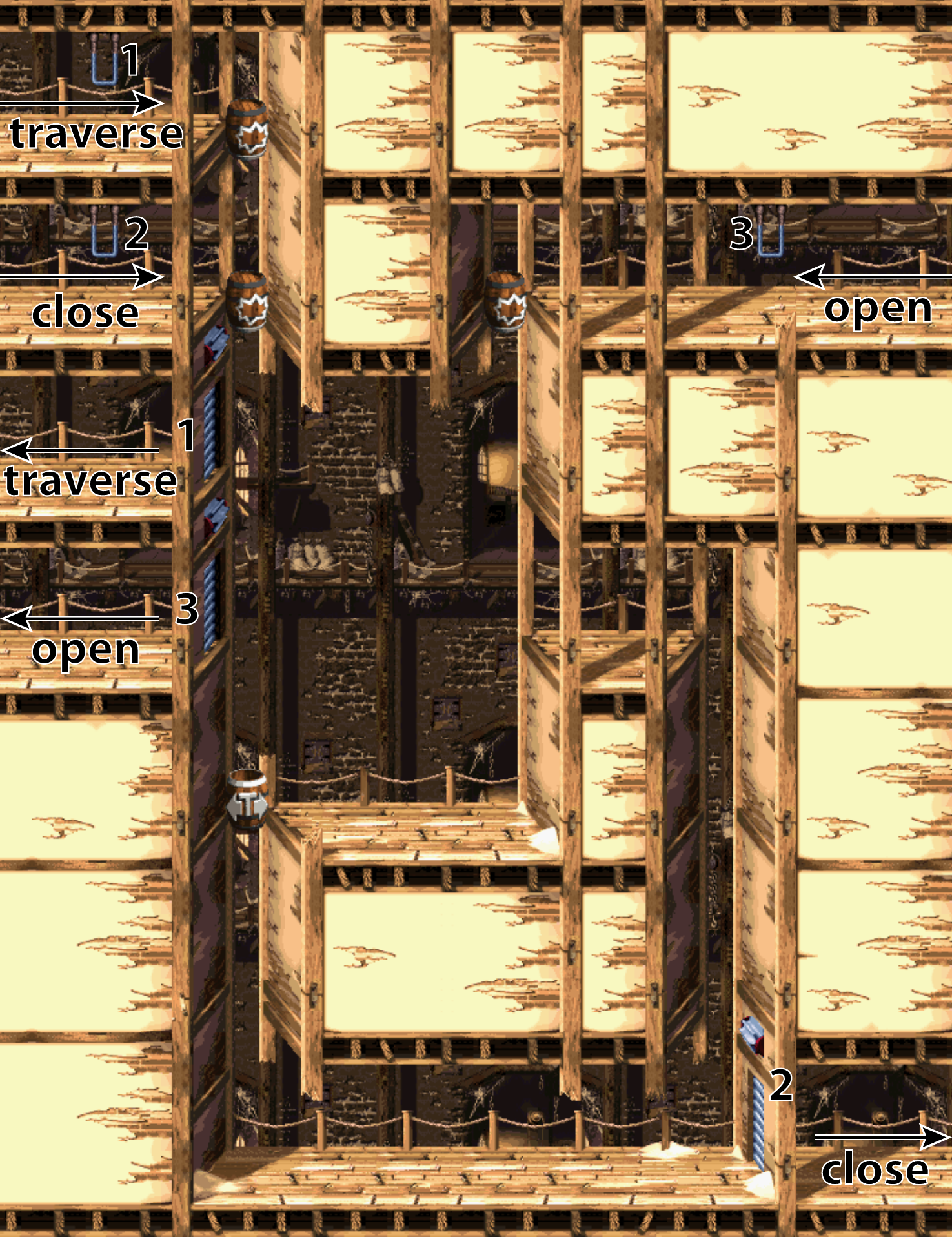}
\caption{Door gadget for Donkey Kong Country 3}
\label{PSPACEDKC3}
\end{figure}

}

\section{The Legend of Zelda}\label{s:zelda}

The Legend of Zelda is a high fantasy flip-screen overhead adventure
video game developed by Nintendo and published for the NES in
1986. Like the Super Mario and Donkey Kong series, the Zelda series
has spawned over a dozen sequels across various gaming platforms. The
series features a young male swordsman named Link, usually clad in a
green tunic. In the original Legend of Zelda game, the player controls
Link and ventures on a quest through dungeons, collecting items and
defeating enemies and bosses to rescue the kidnapped Princess Zelda
from the antagonist,~Ganon. The premise of the SNES sequel, The Legend
of Zelda: A Link to the Past, is roughly the same.
Here we consider both versions, although we do not describe in much detail the
items or enemies found in the games.  Instead, we refer the interested
reader to~\cite{zeldawiki} and~\cite{zeldawiki2}.  Unlike the previous
games we studied, The Legend of Zelda is nonlinear and the player
must acquire new items from dungeons in order to access previously
inaccessible areas. The game world is set in an overworld, Hyrule,
with self-contained dungeons scattered across the world.

We have several different types of proofs, using different elements found
in Zelda games. Pushable blocks play a large role in many puzzles in
Zelda games, and we use these elements to prove that the original
Legend of Zelda is NP-hard and that several later games are
PSPACE-complete. The hookshot is also a recurring item in Zelda games
(although it is absent from the original). With the hookshot, Link
can grapple onto and pull himself towards distant objects. We use this
item to prove that the rest of the Zelda games are NP-hard.
Finally, an element found in most Zelda dungeons are switch-operated doors,
which imply that some Zelda games (most notably, Legend of Zelda: A Link
to the Past) are PSPACE-complete.


\later{
\ifabstract
\else
  \subsection{NP-hardness}
\fi
\label{zelda NP}

\paragraph{Block pushing.}

Generalized Legend of Zelda is NP-hard by reduction from a puzzle
similar to Push-1, because Legend
of Zelda contains blocks which may be pushed according to the same rules as in
Push-1~\cite{Push}, except that in Zelda, each block may be pushed at most once.
Fortunately, all of the gadgets in the reduction for Push-1 found in~\cite{Push}
still function as intended when each block can be pushed at most once, with
the possible exception of the Lock gadget. However, a simple modification
to the Lock gadget (illustrated in Figure~\ref{fig:zelda_lock}) suffices.
(Here we assume that Link has no items, in particular, no raft.)
\begin{figure}[htbp]
\centering
\includegraphics{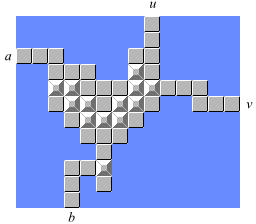}
\caption{Modified Lock gadget for block pushing in Legend of Zelda}
\label{fig:zelda_lock}
\end{figure}
Therefore we obtain
\begin{theorem}\label{t:zelda1}
It is NP-hard to decide whether a given target location is reachable from a
given start location in generalized Legend of Zelda.\hfill\qed
\end{theorem}

Several of the later Zelda games contain ice block puzzles, which will be considered later,
in the PSPACE-completeness section.

\paragraph{Hookshot.}
Unfortunately, the block-pushing reduction does not extend to many games of the
Legend of Zelda series,
because the rules of Push-1 do not allow the player to pull blocks, which is
possible in those games. We now give a proof of NP-hardness for
the SNES sequel, The Legend of Zelda: A Link to the Past, taking advantage
of the hookshot item, which allows Link to grapple onto certain distant objects
and pull himself towards them. Because Link cannot hookshot onto very far objects,
in our reduction we will require him to hookshot at most seven tiles away.
However, all our gadgets keep working even if the hookshot distance is unbounded.

\begin{theorem}\label{t:zelda2}
It is NP-hard to decide whether a given target location is reachable from a given
start location in generalized Legend of Zelda: A Link to the Past.
\end{theorem}
\begin{proof}
The only elements of the game we use are the following. We use
treasure chests and blocks, which will serve as grapple targets.
Link starts out with the hookshot (which can grapple onto
treasure chests and pull Link towards them).

Although this game is not technically a platform game, we use the same
framework provided in Section~\ref{s:platform_framework} here to
reduce from 3-SAT. We do not need any special Start or Finish gadgets.
The Variable gadget, illustrated in Figure~\ref{fig:zelda_var}, works
as follows.  Link approaches from either the top left or top right,
depending on which value of the previous variable was chosen. Then
Link hookshots onto the chest in the top center, and finally hookshots
onto one of the two chests at the bottom. Once Link has reached one of
the two bottom chests, the other one becomes unreachable. Observe how
several barriers around corridors prevent Link from hookshotting
onto the chests from undesirable directions.
%

\begin{figure}[htbp]
\centering
\begin{minipage}{0.45\linewidth}
\centering
\includegraphics{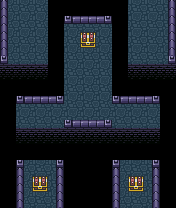}
\caption{Variable gadget for Zelda}
\label{fig:zelda_var}
\end{minipage}\hfill
\begin{minipage}{0.45\linewidth}
\centering
\includegraphics{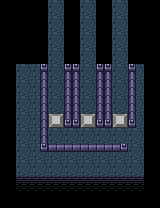}
\caption{Clause gadget for Zelda}
\label{fig:zelda_clause}
\end{minipage}
\vspace{-1ex}
\end{figure}

The Clause gadget is illustrated in Figure~\ref{fig:zelda_clause}. The three
corridors at the top correspond to the literals that appear in the clause.
When Link visits one of these corridors, he may push the block forward,
which allows him to hookshot onto the block from the right later on
while traversing the Check path; see Figure~\ref{fig:zelda_check}. Note
that the leftmost barrier on the Clause gadget prevents Link from
``skipping'' unsatisfied clauses, even if he can hookshot at arbitrarily
long distances.
%

\begin{figure}[htbp]
\centering
\includegraphics[width=\linewidth]{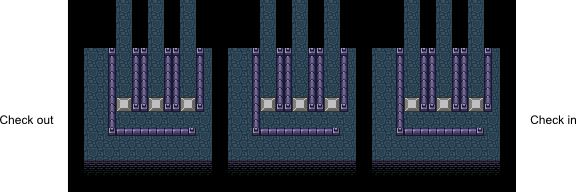}
\caption{Check path for Zelda (with three clauses)}
\label{fig:zelda_check}
\end{figure}

%

Finally, Crossover gadgets are natively implemented in the game, as illustrated in Figure~\ref{fig:zelda_xover_new}.
\end{proof}

\begin{figure}[htbp]
\centering
\includegraphics[width=0.65\linewidth]{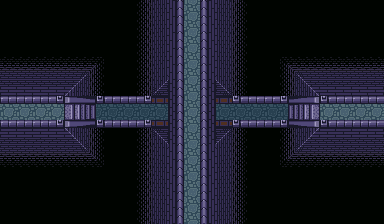}
\caption{Crossover gadget for Zelda}
\label{fig:zelda_xover_new}
\end{figure}

Theorem~\ref{t:zelda2} extends easily to all later games
(Link's Awakening, Ocarina of Time, Majora's Mask, Twilight Princess, etc.)\
because the
hookshot and blocks are present. In the case of the 3D games, we attach
a hookshot target onto the block (as blocks themselves are no longer valid
hookshot targets), and add walls to the sides of platforms, because in 3D games
Link can hookshot towards any visible direction.

\paragraph{Zelda II.}
Zelda II: The Adventure of Link is unique in that it features side-scrolling
levels, with none of the usual mechanics in the top-down ``overworld''
that connects levels together.
In the side-scrolling levels, Link is a $2 \times 1$ character who experiences
gravity and can jump by a bounded height amidst a square grid of blocks and
empty space, similar to Super Mario Bros.
Link can also attack monsters, gain experience, etc.,
but we do not need monsters for proving NP-hardness.
In fact, we just need one main feature---keys and locked doors---as
illustrated in Figure~\ref{fig:zeldaII}.
Keys are identical items that Link can permanently pick up.
Locked doors are impassable, unless Link has a key, in which case
touching the door consumes one key and makes the door permanently disappear.
(Master keys and the Fairy spell are other ways to traverse doors, but we can
 assume they are unavailable to Link in the levels we create.)

\begin{figure}[htbp]
\centering
\includegraphics[scale=1]{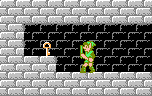}\hfil
\includegraphics[scale=1]{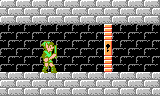}
\caption{Keys and locked doors in Zelda II: The Adventure of Link}
\label{fig:zeldaII}
\end{figure}

\begin{theorem}\label{t:zeldaII}
It is NP-hard to decide whether a given target location is reachable from a given
start location in generalized Legend of Zelda II: The Adventure of Link.
\end{theorem}
\begin{proof}
We apply Metatheorem~3.c from~\cite{Games} because Zelda II keys and
locked doors provide exactly the ``cumulative keys'' model required by the
metatheorem; the metatheorem also requires one-way paths, which we can
build with a big fall as in the Super Mario Bros.\ Variable gadget
(Figure~\ref{fig:mario_var}).
For completeness, the reduction is from Hamiltonian cycle
in directed planar graphs of degree~$3$; hence each vertex has either
two incoming edges and one outgoing edge, or vice versa.
The reduction represents each vertex by an area with two keys, and represents
each directed edge by a one-way path obstructed by a single locked door,
except for one forced edge (e.g., the sole outgoing or sole incoming edge
of any vertex) which is blocked by $|V|$ locked doors.
\end{proof}
}

\subsection{PSPACE-completeness}

\paragraph{Ice blocks.}
Recall that several Zelda dungeons contain ice blocks, which are pushed
like normal blocks, except when pushed they slide all the way until they encounter
an obstacle. These are the same rules as in PushPush1, which is PSPACE-complete,
hence we get the following theorem.

\begin{theorem}\label{t:zelda_pspace}
It is PSPACE-complete to decide whether a given target location is reachable
from a given start location in the generalized versions of the following
games: Ocarina of Time, Majora's Mask, Oracle of Seasons, The Minish Cap,
and Twilight Princess.\hfill\qed
\end{theorem}

\paragraph{Doors and switches.}
Most Zelda dungeons feature switch-operated doors. The following proof works for Legend of Zelda: A Link to the Past, but may be adapted to several other Zelda games. We use the general PSPACE-hardness framework detailed in Section~\ref{s:platform_framework}.

\paragraph{Teleporter tiles.}
In our proof of the following theorem (PSPACE-completeness of LoZ: ALttP)
we use teleporter tiles from the game which, when stepped on, sends Link to
a fixed destination. The teleporter tiles are one-way.

\begin{theorem}
It is PSPACE-complete to decide whether a given target location is reachable from a given start location in generalized Legend of Zelda: A Link to the Past.
\end{theorem}
\begin{proof}
In our generalization, we stipulate that the open-closed state of dungeon Gates that are operated by Floor Switches is ``remembered'' by the game, even if the player exits the room containing such Gates. Floor Switches toggle the open-closed state of some Gates, and Crystal Switches toggle the raised-lowered state of all the Pillars in the dungeon. The main difference between Gates and Pillars is that Gates are always closed when the player enters the dungeon, while Pillars may be either lowered or raised (which means that they are, respectively, traversable and not traversable).

\ifabstract
\else
The Crossover gadget has already been described,
\fi
and the Door gadget is depicted in the lower part of Figure~\ref{PSPACEZelda}.

\begin{figure}[htbp]
\vspace{-0.2in}
\centering
\includegraphics[width=0.65\linewidth]{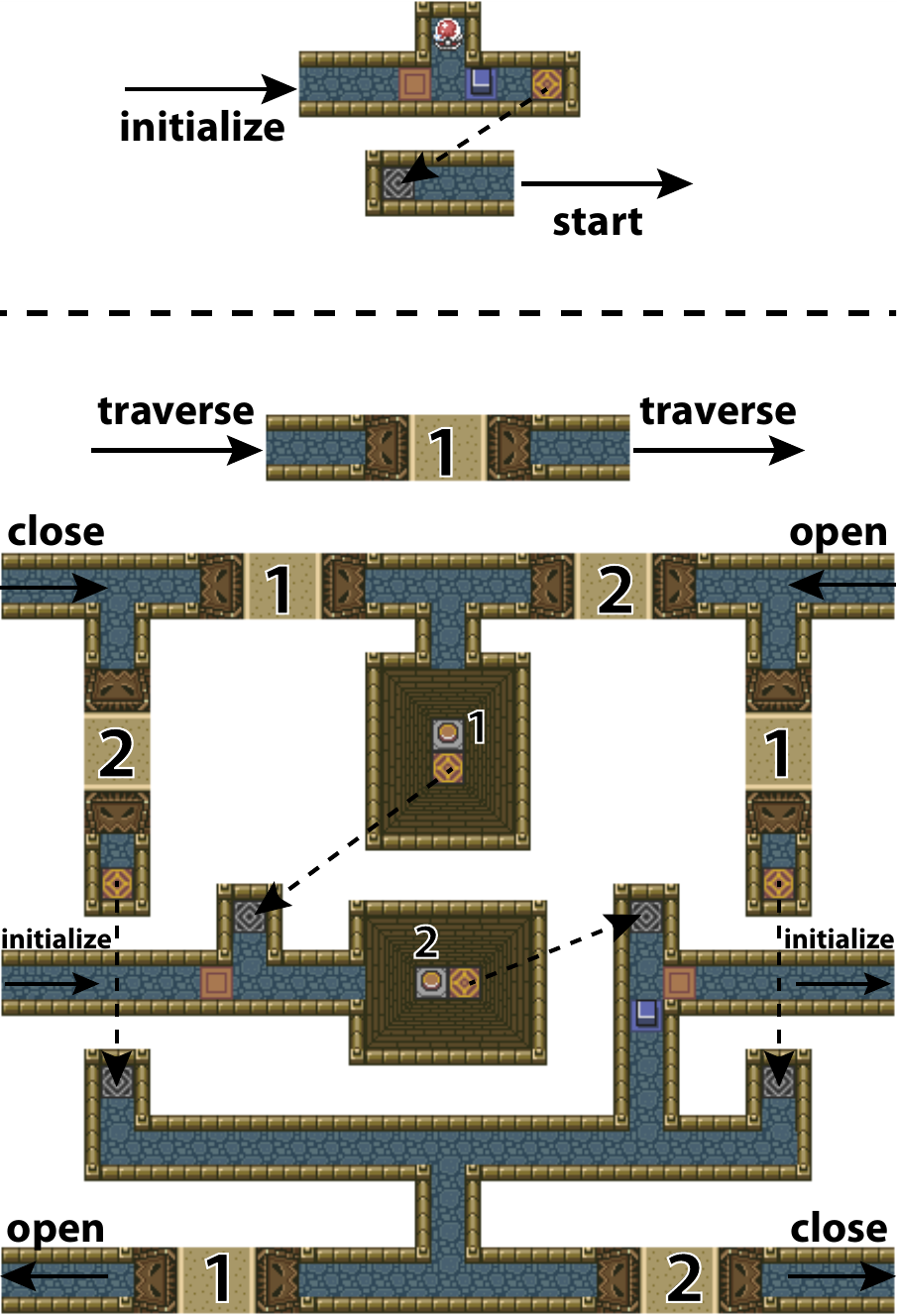}
\caption{Door gadget for Zelda}
\label{PSPACEZelda}
\vspace{-0.2in}
\end{figure}

The general construction is as follows. The player first traverses each Door gadget through its ``initialize'' path, thus pushing Floor Switch number~2 and
opening the corresponding Gates. Floor Switches are placed at the bottom of pits, so that the player must fall on them from above, and then immediately walk in the Teleporter, which implies that the player operates each Floor Switch exactly once per traversal.

After doing this for all Door gadgets, the player reaches the gadget on top of Figure~\ref{PSPACEZelda} and has to operate the Crystal Switch, which toggles the state of every Pillar. Then the player may enter the Teleporter, which is followed by a path that leads to the first Quantifier gadget, and to the rest of the construction described in~\cite[Metatheorem~4.c]{Games}.

From now on, in each Door gadget, the invariant is that Gates number~1 are open if and only if Gates number~2 are closed. The door is considered open if and only if Gates number~1 are open, hence the ``traverse'' path works trivially.

Keeping the invariant in mind, it is easy to see how the ``open'' and ``close'' paths work.  If the door is in the ``wrong'' state, then the player must take the route through the two Floor Switches and toggle the state of every Gate in the gadget.  If the door is in the ``correct'' state already, then the player must take the other route, which avoids the Floor Switches.  Again, each Floor Switch must be pressed exactly once per traversal.

Observe that paths in each Door gadget may also be stretched and deformed, in such a way that all the elements with the same number lie in the same dungeon room.

Membership in PSPACE follows from the remarks in the Introduction.
\end{proof}

\later{
\section{Metroid}\label{s:metroid}

Metroid is an action-adventure video game co-developed by Nintendo
Research and Development~1 and Intelligent Systems, and published for
the NES in 1986. Like the previous games we have described, Metroid
continued to spawn numerous sequels across gaming platforms. The
original Metroid, as well as Super Metroid (SNES), were
two-dimensional side-scrolling flip-screen games. In Metroid, the
player controls Samus, a female space bounty hunter, and explores
planet Zebes to retrieve Metroids (life-draining creatures), which were
stolen by Space Pirates.  Like The Legend of Zelda, Metroid is an
open-ended nonlinear game in which players must discover and acquire
items in order to access previously inaccessible areas. However,
unlike the world of Zelda, the world of Metroid is essentially one
very large overworld with no nicely self-contained dungeons, and hence
is at first glance much less structured.

The mechanics of Metroid are also simple. Samus can walk, jump, or
shoot.  Samus may shoot left, right, or upwards (but not downwards).
In addition, when standing on the ground, Samus can enter Morph Ball mode,
which makes her one tile tall instead of two, and allows her to roll.
However, while in Morph Ball, Samus cannot attack or jump,
but only roll left or right and fall, and eventually morph back into her normal form.
The only enemy in the game we use in our proof is the
Zoomer, which is simply a creature that walks along whatever surface
it adheres to, and deals damage to Samus upon contact, but may be
killed by Samus' normal weapon. If two Zoomers going in opposite directions meet, they traverse each other without colliding. Samus starts with 30 health points, and loses 8 points upon contact with a Zoomer. When a Zoomer is killed it may drop an orb that, if collected by Samus, increases her health by 5 points. However, Samus' health may never increase above 99 points. 
For more details on the game,
we refer the interested reader to~\cite{metroidwiki}.

In this section, we prove the following:
\begin{theorem}\label{t:metroid}
It is NP-hard to decide whether a given target location is reachable
from a given start location in generalized Metroid.
\end{theorem}
\begin{proof}
We use the same framework given by Section~\ref{s:platform_framework}.
In our construction, Samus has acquired no special items yet except
the Morph Ball.  The Start and Finish gadgets are trivial, and the Variable gadget
is the same as in our proofs for Mario and Donkey Kong.

The first nontrivial gadget we implement is the Clause gadget
(Figure~\ref{fig:metroid_clause}). It has literals coming from below,
with the Check path on top. It works as follows.
Samus approaches from below, and can shoot upward
through the gap and kill all the Zoomers (which are moving clockwise
around the platform), freeing up the space for later traversal in Morph
Ball mode. After killing the Zoomers, Samus cannot jump and sneak into
the passageway, because it requires Morph Ball mode. Furthermore, if
no literal is visited, then Samus cannot pass through from above
because the corridor is completely filled with Zoomers. Indeed, if the central platform is large enough, there are enough Zoomers to kill Samus even if she has 99 health points.  Moreover,
Samus cannot traverse from the Check path down to the literal wires
because she would have to jump and morph in mid-air.

\begin{figure}[htbp]
\centering
\includegraphics[scale=0.75]{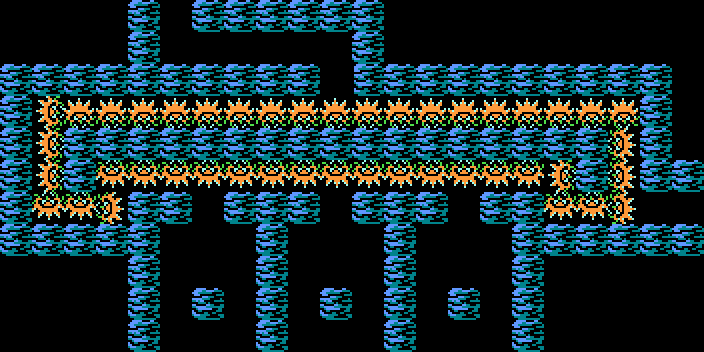}
\caption{Clause gadget for Metroid}
\label{fig:metroid_clause}
\end{figure}

The second nontrivial gadget is the Crossover gadget
(Figure~\ref{fig:metroid_xover}). It can be traversed from the top-left (respectively, top-right) corner to the bottom-right (respectively, bottom-left) corner, without leakage. Hence this Crossover gadget is unidirectional, which suffices by Remark~\ref{r:xover}. Each of the two upper platforms is covered with Zoomers, except for a three-tile-wide area, which moves around the platform as the Zoomers move. The only way Samus can get past an upper platform is to wait above it until the area devoid of Zoomers passes below her. Then she can quickly go down, follow the Zoomers toward the center of the gadget, and fall down onto the lower platform. This platform is traversed by two streams of Zoomers, going in opposite directions, timed in such a way that, if Samus comes from the upper-left (respectively, upper-right) platform, she is forced to go right (respectively, left) to run away from the Zoomers. Note that the three platforms have the same size and shape, hence each Zoomer's cyclic path has the same period. Moreover, if the platforms are large enough and the streams of Zoomers are long enough, there are enough Zoomers to kill Samus if she deviates from her intended behavior, even if she has 99 health points and she keeps jumping as soon as she reaches the bottom platform. Observe that Samus could get in the area just below each of the three platforms, if she sneaks into a lateral passageway while falling in Morph Ball mode. However, after she has done so, she is unable to exit, because she cannot jump in Morph Ball mode or morph in mid-air.

\begin{figure}[htbp]
\centering
\includegraphics[width=\linewidth]{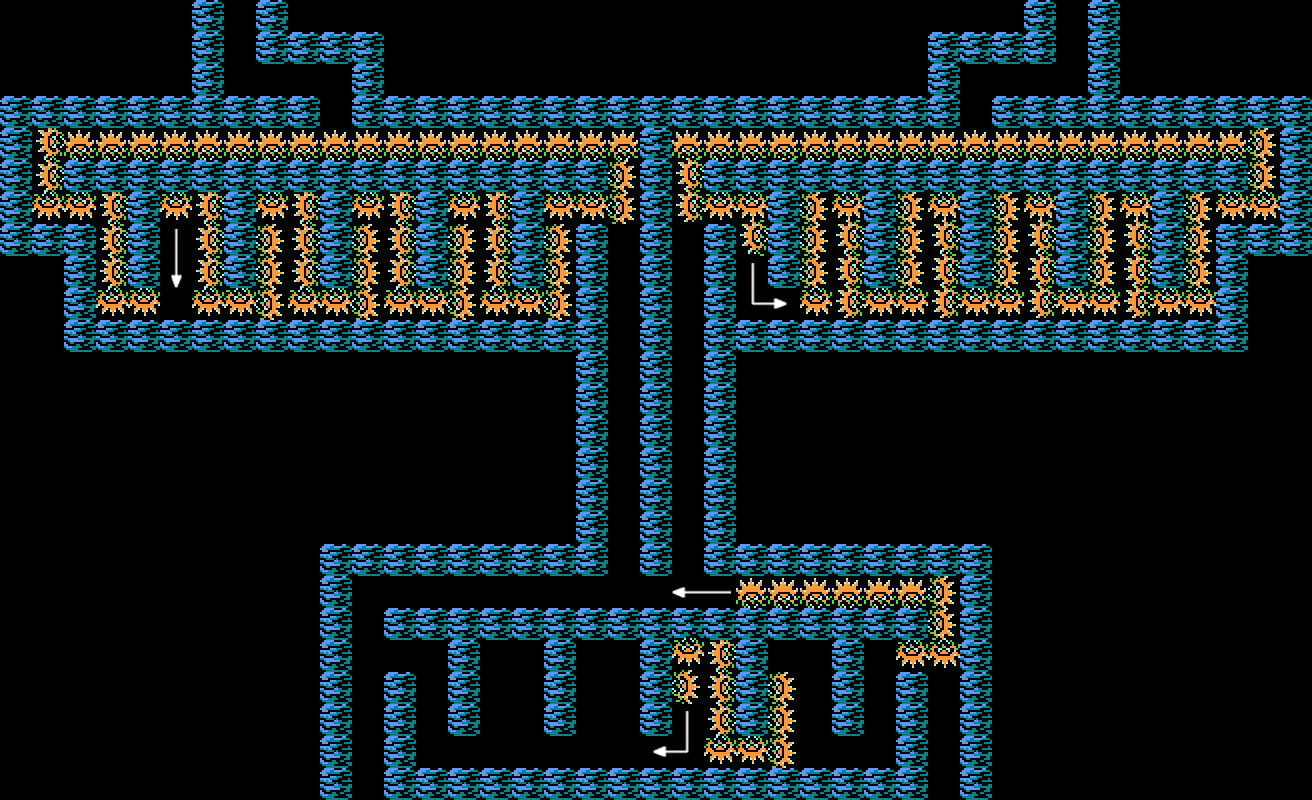}
\caption{Crossover gadget for Metroid}
\label{fig:metroid_xover}
\end{figure}

\end{proof}

\paragraph{Other Metroid games.}
Theorem~\ref{t:metroid} holds just as well for all other
Metroid games, as the Morph Ball and Zoomers are present also in the sequels. However, our gadgets may require some modifications, as in later games (e.g., Super Metroid) Samus can shoot downward and can morph in mid-air. Such modifications are straightforward if additional game elements are used, such as enemies called Rippers, which fly horizontally back and forth and cannot be killed by Samus' normal weapon. For the 3D Metroid games, one can
modify the 2D construction to have very small depth, so that the third
dimension becomes irrelevant.
}

\later{

\section{Pok\'emon}\label{s:pokemon}

Pok\'emon is a series of overhead role-playing games developed by Game
Freak for various handheld Nintendo consoles, including Game Boy, Game
Boy Color, Game Boy Advance, and the Nintendo DS. There are various
versions of Pok\'emon, but the core mechanics of the game are
invariant throughout. The player controls a young teenager and wanders
through the land capturing and training creatures called Pok\'emon
(short for Pocket Monsters). At any one time, the player may hold up
to six Pok\'emon in their team. Each Pok\'emon has a Level, which
indicates roughly how experienced it is, as well as battle stats:
attack, defense, speed, special attack, and special defense.  In
addition, each Pok\'emon has hit points (HP), which indicate how much
damage it can take before ``fainting'', as well as power points (PP)
for each move, which indicate how many times it may use that move.
Pok\'emon battles are two-player simultaneous move games, and in each
round the Pok\'emon with higher speed attacks first. The battle ends
when all of one trainer's Pok\'emon have fainted.

Pok\'emon is NP-hard because it has blocks
that the player can push according to the rules of Push-1
\cite{Push}. Therefore we immediately have

\begin{theorem}\label{t:pokemon1}
It is NP-hard to decide whether a given target location is reachable from a
given start location in generalized Pok\'emon.\hfill\qed
\end{theorem}

In this section, we give an alternate proof that Pok\'emon is NP-hard
using no elements of the game except enemy Trainers and the game's battle
mechanics (and thus no blocks).
Therefore we obtain the following stronger result:

\begin{theorem}\label{t:pokemon2}
It is NP-complete to decide whether a given target location is reachable
from a given start location in generalized Pok\'emon in which the only
overworld game elements are enemy Trainers.
\end{theorem}
\begin{proof}
We briefly describe the mechanism for enemy Trainer encounters in the
Pok\'emon games. Each enemy Trainer has an orientation (which
direction they are facing), a range of sight, and a set of
Pok\'emon. If the player walks into the line of sight of the Trainer
(and if such a line of sight is not occluded by some other element, e.g., another Trainer),
the player is forced to stop, the Trainer walks toward the player
until they are adjacent, and then battle ensues.  Additionally, if
the player approaches a Trainer from outside the Trainer's line of
sight (i.e., from behind or the sides), the player may talk to the Trainer,
activating the battle.  If the player wins the battle, the Trainer will
not move or attack again.

We prove NP-hardness by using the framework in Section~\ref{s:platform_framework}. Start and Finish gadgets are trivial, hence it suffices to implement the Variable, Clause, and Crossover gadgets.

%
In our implementations, we use three kinds of objects. Walls,
represented by dark grey blocks, cannot be occupied or walked
through.
Trainers' lines of sight are indicated by translucent rectangles.
We have two types of Trainers. Weak Trainers, represented by red rectangles,
are Trainers whom the player can defeat with certainty without expending
any effort, i.e., without consuming PP or taking damage.  Strong
Trainers, represented by blue rectangles, are Trainers against whom the player
will always lose.

We can implement weak and strong Trainers as follows
(a construction due to Istvan Chung).
Weak Trainers each hold a Level 100 Electrode with maximum Speed
and equipped with only the Self Destruct move.
Strong Trainers each hold two Snorlaxes, with Speed of 30.
The player has
no items, and only one Pok\'emon in his team. For Generation I and II games
(Red/Blue/Yellow and Gold/Silver/Crystal versions respectively), the player holds
a Gastly which has learned Self Destruct using TM36, and its PP for its other moves
have all been expended, so it can only use Self Destruct in battle.
When the player encounters a weak Trainer, the enemy Electrode will move first and
use Self Destruct, which deals no damage to Gastly because Self Destruct is a Normal type attack
and Gastly is Ghost type, so the weak Trainer immediately loses.
When the player encounters a strong Trainer, Gastly moves first and uses Self Destruct, causing
the player to lose (even if it defeats the enemy Snorlax, the opponent holds another one).
This implementation only works in Generations I and II because TM36 exists only in Generation I
and the Time Capsule feature in Generation II allows a Gastly with Self Destruct to be traded
from Generation I to Generation II. In Generations III, IV, and V, Gastly can be replaced by Duskull,
which is allowed to learn the move Memento, which serves the same purpose as Self Destruct, via
breeding.

We now implement each gadget. The Variable gadget, illustrated in Figure~\ref{fig:pokemon_variable},
must force the player to choose either an $a$-to-$b$ traversal
or an $a$-to-$c$ traversal. We show that this is the case.
The player enters through $a$. If the player wants to
exit through $b$, they walk into the Trainer's line of sight, luring
the Trainer down and opening up the passage to $b$, but closing the
passage to $c$. On the other hand, if the player wants to exit through
$c$, they walk up to the Trainer and talk to him, disabling the
Trainer, so that they may walk around and exit through $c$.

\begin{figure}[htbp]
\centering
\includegraphics[scale=0.5]{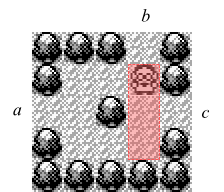}
\caption{Variable gadget for Pok\'emon}
\label{fig:pokemon_variable}
\end{figure}

The Clause gadget, depicted in Figure~\ref{fig:pokemon_clause}, is unlocked whenever the player enters from one of the top paths and talks to one of the three weak Trainers, disabling him. When the Check path is traversed, every weak guard who has not been disabled is lured down by the player, thus getting out of the line of sight of the rightmost strong Trainer. Eventually the player must cross that line of sight, and is reached by the strong Trainer if and only if no weak Trainer blocks him, which happens if and only if the Clause gadget is still locked. The leftmost strong Trainer prevents leakage from the Check path to the literals, in case the leftmost weak Trainer has been lured down.

\begin{figure}[htbp]
\centering
\includegraphics[scale=0.5]{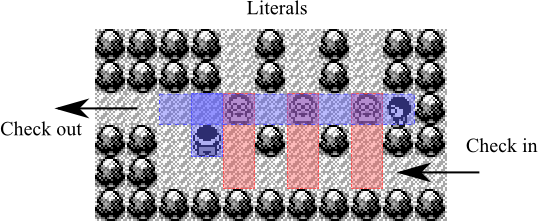}
\caption{Clause gadget for Pok\'emon}
\label{fig:pokemon_clause}
\end{figure}

Before presenting the Crossover gadget, we introduce the \emph{Single-use path}, a path that can only be traversed once, and only in one direction. This is implemented by the gadget in
Figure~\ref{fig:pokemon_nr}, which can be traversed only once, from $a$ to $b$.  Clearly, the player cannot enter via
$b$, because that lures the weak Trainer to block the passage. Suppose
the player enters through $a$. They can safely walk to $b$, because the
weak Trainer is blocking the bottom strong Trainer's line of
sight. However, to reach $b$, the player must lure the weak Trainer
out of the line of sight of the strong Trainer, hence the player may
never return in the reverse direction. Also, there is a strong Trainer
above the weak Trainer, preventing the player from disabling the weak Trainer
by entering from $a$.

\begin{figure}[htbp]
\centering
\includegraphics[scale=0.5]{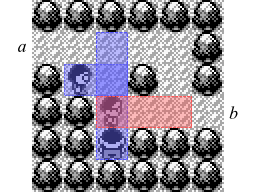}
\caption{Single-use path for Pok\'emon}
\label{fig:pokemon_nr}
\end{figure}

Finally, the Crossover gadget is shown in Figure~\ref{fig:pokemon_xover}. Each of its two passages, $x$-to-$x'$ and $y$-to-$y'$, is unidirectional and may be traversed only once, which suffices due to Remarks~\ref{r:xover} and~\ref{r:xover2}.

Because the Crossover gadget is symmetric, we assume, without loss of generality, that the player attempts to traverse the $x$-to-$x'$ passage first. Upon entering from $x$, they go down and disable the bottom weak Trainer, then go back up and proceed to the right. Observe that the top-left part of the gadget is a ``crossroads'' made of four isometric copies of the Single-use path of Figure~\ref{fig:pokemon_nr}. Upon reaching it, the player is forced to go either down or right. If they go down, they eventually get stuck, because the crossroads is now unreachable due to the just traversed Singe-use path, and the $y'$ exit is unreachable too, because the top weak Trainer is lured all the way down to block the player, as soon as they attempt to exit. On the other hand, if the player proceeds to the right upon entering the crossroads, they can safely reach the $x'$ exit, because the leftmost weak Trainer has been previously disabled. Note that, by doing so, the player incidentally crosses the line of sight of the top weak Trainer, luring him down and disabling him.

Now, if the player wishes to traverse the $y$-to-$y'$ passage, they must take the vertical Single-use paths of the crossroads, because the other two have already been traversed. At this point, the player has no choice but to exit from $y'$, which is safely reachable because the top weak Trainer has already been disabled during the first traversal of the gadget, and is not blocking the way out.

\begin{figure}[htbp]
\centering
\includegraphics[scale=0.5]{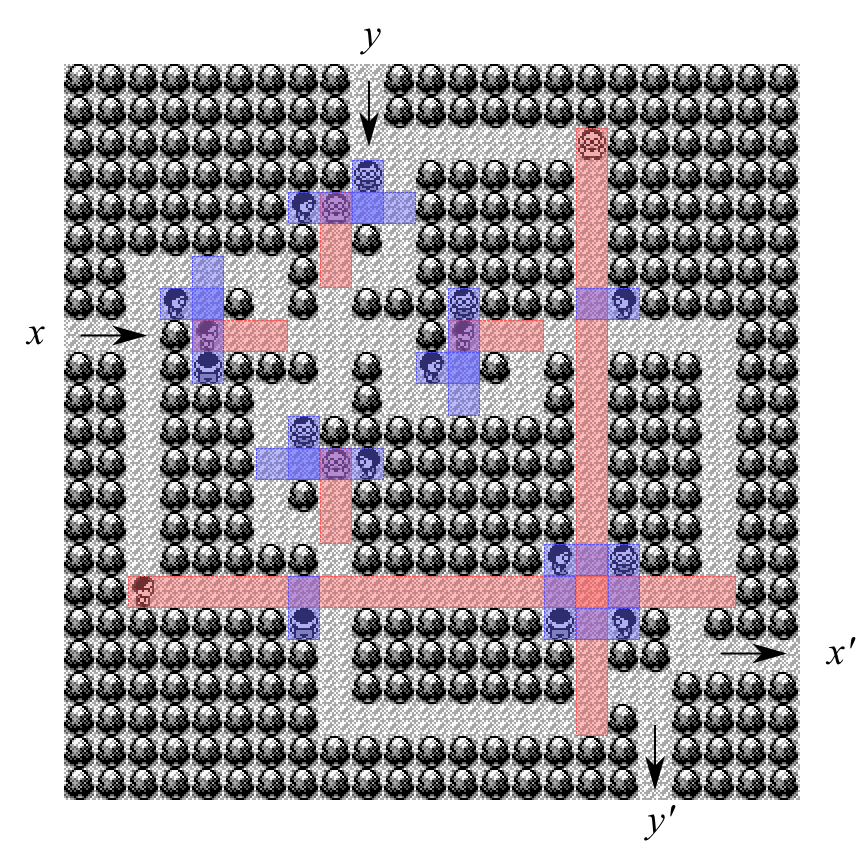}
\caption{Crossover gadget for Pok\'emon}
\label{fig:pokemon_xover}
\end{figure}

\paragraph{NP.}
To show that Pok\'emon with only enemy Trainers is in NP, note that
once a Trainer has been battled, they become inert. Moreover, each
actual battle is bounded in length by a constant, because eventually all
Pok\'emon must expend all their PP for their moves and use Struggle,
a standard physical attack which hurts the opponent as well as the user.
Moreover, each Pok\'emon only has four different attacks, so the branching
factor of the game tree for each battle is constant (and hence the size
of the game tree is also constant).
Therefore, the player may
nondeterministically guess a solution path through the overworld, and
for each battle compute the winning strategy in constant time.
\end{proof}

\paragraph{Other Pok\'emon games.}
Theorem~\ref{t:pokemon2} actually holds for all
Pok\'emon role-playing games, because the only actual Pok\'emon used in
the construction were from the first generation, which is present in
all the games, and enemy Trainers are of course present in all the
games as well.

}



\section*{Acknowledgments}

This work was initiated at the 25th Bellairs Winter Workshop on Computational 
Geometry, co-organized by Erik Demaine and Godfried Toussaint, held on
February 6--12, 2010, in Holetown, Barbados.  We thank the other participants 
of that workshop---Brad Ballinger, Nadia Benbernou, Prosenjit Bose,
David Charlton, S\'ebastien Collette, Mirela Damian, Martin Demaine,
Karim Dou\"{\i}eb, Vida Dujmovi\'c, Robin Flatland, Ferran Hurtado,
John Iacono, Krishnam Raju Jampani, Stefan Langerman, Anna Lubiw, Pat Morin,
Vera Sacrist\'an, Diane Souvaine, and Ryuhei Uehara---for providing a
stimulating research environment.  In particular, Nadia Benbernou was
involved in initial discussions of Super Mario Bros.

We thank readers Bob Beals, Curtis Bright, Istvan Chung,
Peter Schmidt-Nielsen, Patrick Xia, and the anonymous referees for helpful
comments and corrections, and for ``beta-testing'' our constructions.

We also thank The Spriters Resource~\cite{spriters}, VideoGameSprites~\cite{vgsprites},
NES Maps~\cite{nesmaps}, and SNES Maps~\cite{snesmaps} for serving as
indispensable tools for providing easy and comprehensive access to the sprites used
in our figures.

\iffull
Greg Aloupis would like to thank James Lamarche for countless hours of
cooperative research on Nintendo games over the years.
\fi

Finally, of course, we thank Nintendo and the associated developers
for bringing these timeless classics to the~world.

\appendix
\magicappendix

\end{document}